\documentclass[12pt,draftcls, onecolumn,journal]{IEEEtran}
\usepackage{times}
\usepackage[reqno]{amsmath}
\usepackage{amsfonts}
\usepackage{caption3}
\usepackage{times,amsmath,epsfig}
\usepackage{latexsym,amssymb}
\usepackage{rotating,color}
\usepackage{cite}
\usepackage{extarrows}
\usepackage{amsthm}

\def\myQED{\mbox{\rule[0pt]{1.5ex}{1.5ex}}}

\DeclareMathOperator{\essinf}{essinf}

\newcommand{\phiv}{\hbox{\boldmath$\phi$}}

\newcommand{\no}{\nonumber}

\newtheorem{thm}{Theorem}[section]

\newtheorem{prop}[thm]{Proposition}
\newtheorem{lem}[thm]{Lemma}
\newtheorem{cor}[thm]{Corollary}

\newtheorem{rmk}[thm]{Remark}
\newtheorem{exmpl}[thm]{Example}
\IEEEoverridecommandlockouts

\makeatletter
\def\ps@headings{%
\def\@oddhead{\mbox{}\scriptsize\rightmark \hfil \thepage}%
\def\@evenhead{\scriptsize\thepage \hfil \leftmark\mbox{}}%
\def\@oddfoot{}%
\def\@evenfoot{}}
\makeatother
\begin{document}

\title{Quickest Search over Multiple Sequences with Mixed Observations}

\author{Jun Geng, Weiyu Xu and Lifeng Lai\thanks{The work of J. Geng and L. Lai was supported by the National Science Foundation under grant DMS-12-65663. The results in this paper were presented in part at IEEE International Symposium on Information Theory, Istanbul, Turkey, July, 2013.

J. Geng and L. Lai are with the Department of Electrical and Computer Engineering, Worcester Polytechnic Institute, Worcester, MA, 01609, USA (Emails: \{jgeng, llai\}@wpi.edu).

W. Xu is with the Department of Electrical and Computer Engineering, University of Iowa, Iowa City, IA, 52242, USA (Email:weiyu-xu@uiowa.edu).}}
\maketitle 



\begin{abstract}
The problem of sequentially finding an independent and identically distributed (i.i.d.) sequence that is drawn from a probability distribution $f_1$ by searching over multiple sequences, some of which are drawn from $f_1$ and the others of which are drawn from a different distribution $f_0$, is considered. The observer is allowed to take one observation at a time. It has been shown in a recent work that if each observation comes from one sequence, the cumulative sum test is optimal. In this paper, we propose a new approach in which each observation can be a linear combination of samples from multiple sequences. The test has two stages. In the first stage, namely scanning stage, one takes a linear combination of a pair of sequences with the hope of scanning through sequences that are unlikely to be generated from $f_1$ and quickly identifying a pair of sequences such that at least one of them is highly likely to be generated by $f_1$. In the second stage, namely refinement stage, one examines the pair identified from the first stage more closely and picks one sequence to be the final sequence. The problem under this setup belongs to a class of multiple stopping time problems. In particular, it is an ordered two concatenated Markov stopping time problem. We obtain the optimal solution using the tools from the multiple stopping time theory. The optimal solution has a rather complex structure. For implementation purpose, a low complexity algorithm is proposed, in which the observer adopts the cumulative sum test in the scanning stage and adopts the sequential probability ratio test in the refinement stage. The performance of this low complexity algorithm is analyzed when the prior probability of $f_{1}$ occurring is small (referred to as $f_{1}$ being rare). Both analytical and numerical simulation results show that this search strategy can significantly reduce the searching time when $f_{1}$ is rare. The proposed two stage mixed observation strategy can also be easily extended to multiple stages.
\end{abstract}


\section{Introduction} \label{sec:intro}
The quickest search over multiple sequences problem, a generalization of the classical sequential hypothesis testing problem~\cite{Wald:AMS:45}, is originally proposed in a recent paper~\cite{Lai:TIT:11}. In particular, the authors consider a case that multiple sequences are available. For each individual sequence, it may either be generated by distribution $f_{0}$ or $f_{1}$, and its distribution is independent of all other sequences. An observer can take observations from these sequences. The observations taken from the same sequence are identical and independently distributed (i.i.d.). The goal is to find a sequence that is generated by $f_{1}$ as quickly as possible under an error probability constraint. Assuming that the observer can take one observation from a {\it single} sequence at a time and no switch-back is allowed, \cite{Lai:TIT:11} shows that the cumulative sum (CUSUM) test is optimal. This quickest search problem has applications in various fields such as cognitive radio \cite{Jiang:GLOBE:08,Li:CISS:08} 
and database search. The sample complexity of a such search problem is analyzed in~\cite{Malloy:TIT:12}.~\cite{Bayraktar:12} studies the search problem over continuous time Brownian channels. \cite{Tajer:ASIL:13} proposes an adaptive block sampling strategy for the quick search problem and solves the problem for Gaussian signals. The problem of recovering more than one sequence generated from $f_1$ is considered in~\cite{Tajer:ALL:12}.

In this paper, we propose a new search approach, namely a mixed observation search strategy, to quickly find a sequence generated by $f_{1}$. This search strategy consists of two stages. In the first stage, namely the scanning stage, the observer takes observations that are linear combinations of samples from two different sequences. In certain applications, such as cognitive radios, it is easy to obtain an observation that is a linear combination of signals from different sequences. The purpose of this stage is to scan through sequences generated by $f_0$ and quickly identify a pair of sequences among which at least one of them is highly likely to be generated by $f_{1}$. In particular, if the observer believes that both sequences that generate the observation are from $f_{0}$, then it discards this set of sequences and switches to observe another two new sequences. Otherwise, the observer stops the scanning stage and enters the refinement stage. In the refinement stage, the observer examines the two candidate sequences identified in the scanning stage one by one, and makes a final decision on which one of these two sequences is generated by $f_{1}$. Hence, in the refinement stage, no mixing is used anymore.

The motivation to propose this mixed observation search strategy is to improve the search efficiency when the presence of $f_{1}$ is rare. If most of the sequences are generating by $f_{0}$, then the observer can scan through and discard the sequences more quickly by this mixed strategy. Our strategy has a similar flavor with that of the group testing~\cite{Dorfman:AmS:43} and compressive sensing~\cite{Donoho:TIT:06, Candes:TIT:061} in which linear combinations of signals are observed.

With this mixed observation strategy, our goal is to minimize the average search delay under a false identification error constraint. By Lagrange multiplier, this problem is equivalent to minimize a linear combination of the search delay and the false identification error probability. Toward this goal, we optimize over four decision rules: 1) the stopping time for the scanning stage $\tau_{0}$, which determines when one should stop the scanning stage and enter the the refinement stage; 2) the sequence switching rule in the scanning stage $\phiv$, which determines when one should switch to new sequences for scanning; 3) the stopping time for the refinement stage $\tau_{1}$, which determines when one should stop the whole search process; and 4) the final decision rule in the refinement stage $\delta$, which determines which sequence will be claimed to be generated from $f_1$. Figure \ref{fig:system} illustrates this search strategy. This two stage search problem can be formulated as an optimal multiple stopping time problem, which is studied very recently in ~\cite{Carmona:MF:08, Carmona:MOR:08, Kobylanski:AAP:11, Christensen:SADA:13}. In particular, we show that this problem can be converted into an ordered two concatenated Markov stopping time problems. Using the optimal multiple stopping time theory~\cite{Kobylanski:AAP:11}, we derive the optimal strategy for this search problem. We show that the optimal solutions of $\tau_{0}$ and $\phiv$ are region rules. The optimal solution for $\tau_{1}$ is the time when the cost of the false identification is less than the future cost, and the optimal decision rule $\delta$ is to pick the sequence with a larger posterior probability of being generated by $f_{1}$.

\begin{figure}[thb]
\centering
\includegraphics[width=0.5 \textwidth]{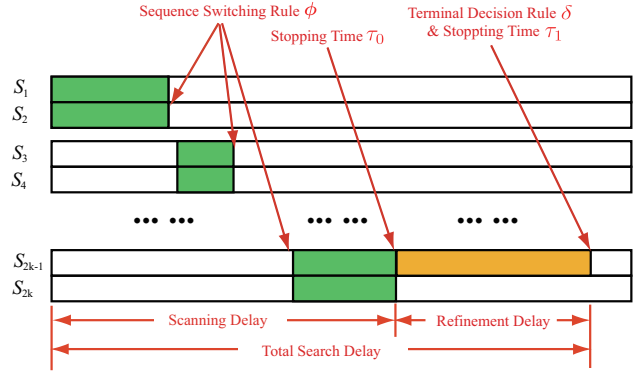}
\caption{Two stage search strategy }
\label{fig:system}
\end{figure}

Unfortunately, the optimal solution for this mixed observation search strategy has a very complex structure. For implementation purpose, we further propose a low complexity search algorithm in which the observer adopts the CUSUM algorithm in the scanning stage and adopts the sequential probability ratio test (SPRT) in the refinement stage. The asymptotic performance of this low complexity algorithm when $f_{1}$ is rare (i.e. the probability that a sequence is generated by $f_{1}$ is small) indicates that the average search delay of the mixed search strategy is dominated by the search delay in the scanning stage, which further depends on the undershoot of CUSUM crossing the lower bound. We analytically derive the reduction of the search delay of the proposed low complexity mixed observation search strategy over the single observation search strategy in \cite{Lai:TIT:11}. For general $f_{0}$ and $f_{1}$, the reduction of the detection delay needs to be computed numerically. For some special pdfs, we show that the search delay of the mixed observation search strategy is half of that of the approach in \cite{Lai:TIT:11}.

This mixed observation search strategy can be easily extended to multiple stages. In particular, we consider a simple extension that the observer's observation is mixed from $2^{K}$ sequences in the scanning stage.  The observer enters the refinement stage when it believes that at least one of the currently observing sequences is generated from $f_{1}$. In the refinement stage, the observer conducts binary search recursively. Specifically, in the $i^{th}$ recursion, the observer divides the  $2^{K-i+1}$ candidate sequences equally into two groups, and sequentially takes the observations mixed from the $2^{K-i}$ sequences in one group. Based on these sequential observations, the observer infers which of the two groups contains the sequence generated from $f_{1}$, and selects one of the two groups for the next recursion. Hence the observer can claim a sequence to be generated from $f_{1}$ after $K$ recursions. This procedure can be modeled as a multiple stopping time problem with $K+1$ stopping times. Similar to the two observation mixed strategy, this problem can be converted to $K+1$ concatenated single stopping time problems, and can be solved using the tools from the optimal stopping theory. We comment that, when $K$ increases, the complexity of the problem grows exponentially. Furthermore, when $K$ is large, it becomes more difficult to distinguish different hypothesis as they are becoming closer. Finding the optimal number of $K$ is subject to future study.

The remainder of the paper is organized as follows. Section \ref{sec:single} briefly reviews the quickest search problem proposed in \cite{Lai:TIT:11} and conducts an asymptotic analysis for the optimal solution when $f_{1}$ is rare. Section \ref{sec:mixed} formulates the mixed search strategy and presents the corresponding optimal solution. A low complexity mixed search algorithm is proposed and its asymptotic performance is analyzed in Section \ref{sec:mixed_asym}.  Section~\ref{sec:extension} extends the mixed search strategy to multiple sequences. Numerical examples are provided in Section~\ref{sec:simulation}. Finally, Section~\ref{sec:conclusion} offers concluding remarks.

\section{Quickest Search over Multiple Sequences with Single Observation  Strategy} \label{sec:single}
\subsection{Review: Problem Formulation and Optimal Strategy} \label{subsec:single_model}
In this subsection, we briefly review the sequential quickest search problem studied in \cite{Lai:TIT:11}. Consider an infinite number of sequences indexed by $s=1, 2, \ldots$. For each sequence $s$, which is denoted as $\{Y_{k}^{s}, k=1, 2, \ldots \}$, the distribution of its samples obeys one of the following two hypotheses:
\begin{eqnarray}
&&H_{0}:\; Y_{k}^{s} \; \overset{i.i.d.}{\sim} \; P_{0}, \; k=1,2,\cdots\no\\
&&H_{1}:\; Y_{k}^{s} \; \overset{i.i.d.}{\sim} \; P_{1}, \; k=1,2,\cdots\no
\end{eqnarray}
where $P_{0}$ and $P_{1}$ are two distinct probability measures that are absolutely continuous to each other. Let $f_0$ and $f_1$ be the probability density functions (pdf) of $P_{0}$ and $P_{1}$, respectively. The goal is to find a sequence generated by $f_{1}$ as quickly and reliably as possible. Moreover, each sequence is assumed to be independent of all other sequences, and each sequence is generated by $f_{1}$ with prior probability $\pi_{0}$ and by $f_0$ with $1-\pi_{0}$.

\cite{Lai:TIT:11} proposes a \emph{single observation} search strategy. At each time slot, the observer takes an observation from one sequence. Then, the observer has to make one of the following three actions: 1) stops the whole search procedure and claims that the currently observing sequence is generated by $f_{1}$; 2) continues taking observation from the same sequence to gather more information about its statistical nature; or 3) abandons the currently observing sequence and switches to observe a new sequence. In \cite{Lai:TIT:11}, it is assumed that if a sequence has been abandoned, the observer will not come back to test it again.

Denote the observation sequence as $\{Z_{k}\}$ with
$$Z_{k}=Y_{k}^{s_{k}},$$
where $s_{k}$ is used to denote the index of the sequence that the observer observes at time $k$. The observations generate the filtration $\{ \mathcal{F}_{k}, k=1,2, \ldots \}$ with
$$\mathcal{F}_{k} = \sigma\{Z_{1}, \ldots, Z_{k}\}. $$
Let $\phi_{k}$ be the sequence switching indicator function at time slot $k$. Specifically, $\phi_{k}$ is a $\mathcal{F}_{k}$ measurable function valued in $\{0, 1\}$. $\phi_{k}=1$ indicates that the observer abandons the currently observing sequence and switches to observe the next sequence, that is $s_{k+1}=s_{k}+1$, while $\phi_{k}=0$ indicates that the observer continues observing the sequence $s_{k}$, that is $s_{k+1}=s_{k}$. Denote $\phiv = \{\phi_{1}, \phi_{2}, \ldots, \phi_{k}\}$. Let $\tau$ be the time that the observer declares that the currently observing sequence is generated by $f_{1}$. Hence $\tau$ is a stopping time adapted to $\{ \mathcal{F}_{k} \}$.

Two performance metrics, namely the average search delay (ASD) and the false identification probability (FIP), are of interest. ASD is defined as
$$ ASD = \mathbb{E}[\tau],$$
and FIP is defined as
$$ FIP = P(H^{s_{\tau}} = H_{0}). $$
Here, ASD and FIP are defined with respect to the probability measure $P_{\pi_{0}} = \pi_{0}P_{1}+(1-\pi_{0})P_{0}$. For the brevity of notations, in this paper we use $P(\cdot)$ and $P_{\pi_{0}}(\cdot)$($\mathbb{E}[\cdot]$ and $\mathbb{E}_{\pi_{0}}[\cdot]$) interchangeably. The latter notation is used when we want to emphasize the prior probability $\pi_{0}$. The goal in \cite{Lai:TIT:11} is to find a sequence of switching rules $\phiv$ and a stopping time $\tau$ that jointly minimize ASD subjected to a FIP constraint. Specifically, \cite{Lai:TIT:11} wants to solve the following optimization problem:
\begin{eqnarray}
\inf_{\tau, \phiv} ASD \quad  \text{ subject to } \quad FIP \leq \zeta. \label{eq:single_problem}
\end{eqnarray}

It is shown in \cite{Lai:TIT:11} that $\pi_{k} := P(H^{s_{k}} = H_{1}|\mathcal{F}_{k})$ is a sufficient statistic for this problem. Using Bayes' rule, $\pi_{k}$ can be shown to satisfy the recursion
\begin{eqnarray}
\pi_{k+1} &=& \frac{\pi_{k}f_{1}(Z_{k+1})}{\pi_{k}f_{1}(Z_{k+1}) + (1-\pi_{k})f_{0}(Z_{k+1})}\mathbf{1}_{\{\phi_{k}=0\}} \no\\
&+& \frac{\pi_{0}f_{1}(Z_{k+1})}{\pi_{0}f_{1}(Z_{k+1}) + (1-\pi_{0})f_{0}(Z_{k+1})}\mathbf{1}_{\{\phi_{k}=1\}}, \label{eq:single_stat}
\end{eqnarray}
where $\mathbf{1}_{\{\cdot\}}$ is the indicator function. In \cite{Lai:TIT:11}, the optimal solution is shown to be
\begin{eqnarray}
&& \phi_{k}^{*} = \left\{ \begin{array}{cc}
             1 & \text{if } \pi_{k} < \pi_{0} \\
             0 & \text{otherwise}
           \end{array} \right., \no \\
&& \tau^{*} = \inf \{ k \geq 0 | \pi_{k} > \pi_{U}^{*} \}, \label{eq:single_opt}
\end{eqnarray}
where $\pi_{U}^{*}$ is a constant such that the FIP constraint holds with equality.

The optimal solution in \eqref{eq:single_opt} can be shown to be equivalent to the CUSUM test. In particular, the observer switches to observe a new sequence when a reset occurs in the corresponding CUSUM, and the observer terminates the search procedure when  the corresponding CUSUM terminates.

\subsection{Asymptotic Performance of the Single Observation Search Strategy} \label{subsec:single_performance}
In this subsection, in preparation for the analysis of the proposed scheme in Section \ref{sec:mixed}, we study the asymptotic performance of the CUSUM algorithm described in \eqref{eq:single_stat} and \eqref{eq:single_opt}. As mentioned in the introduction, the sampling complexity of the multiple sequence search problem has been discussed in \cite{Malloy:TIT:12}. In particular, \cite{Malloy:TIT:12} considers a general setting (i.e., $\pi_{0}$ is not necessarily small) and presents an upper bound and a lower bound of ASD for the CUSUM search algorithm. The upper and lower bounds derived in \cite{Malloy:TIT:12} are both inversely proportional to the Kullback-Leibler (KL) divergence of $f_{0}$ and $f_{1}$, but with \emph{different coefficients}. However, this is not accurate enough for our further analysis. In this section, we focus only on the case that $H_{1}$ is \emph{rare} (i.e., $\pi_{0}$ is small) and we derive the asymptotic value of ASD by the connection between the renewal process and CUSUM.

To endow the stopping time with the general sense, we rewrite it as
$$\tau = \inf \{ k \geq 0 | \pi_{k} > \pi_{U} \}.$$
That is, we replace $\pi_{U}^{*}$, which is optimal only for a particular $\zeta$, by a general $\pi_{U} \in [0, 1)$.

The CUSUM algorithm can be viewed as a renewal process with each renewal occurring whenever $\pi_{k}$ is reset to $\pi_{0}$, and with a termination whenever $\pi_{k}$ exceeds the upper bound $\pi_{U}$. Therefore
$$\tau = \sum_{n=1}^{N} \eta_{n},$$
where $\eta_{1}, \ldots, \eta_{n}, \ldots$ are i.i.d. repetitions of $\eta = \inf\{k\geq0|\pi_{k} \notin [\pi_{0}, \pi_{U}]\}$ under $P_{\pi_{0}}$, and $N$ is the number of repetitions. Notice that the observer switches to observe a new sequence when $\pi_{k}$ is reset to $\pi_{0}$ and terminates the search process when $\pi_{k}$ exceeds $\pi_{U}$, hence $\eta_{n}$ can be interpreted as the time spent in the $n^{th}$ sequence and $N$ can be viewed as the total number of sequences searched.

The CUSUM test can be represented equivalently in terms of likelihood ratios (LR). Notice that the observer keeps observing the same sequence from $\eta_{n}+1$ to $\eta_{n+1}$. Since $\eta_{n}$'s are i.i.d repetitions of $\eta$, we can focus only on the evolution of $\pi_{k}$ for $k = 1, \ldots, \eta_{1}$. 
We have
\begin{eqnarray}
\pi_{k} = \frac{\pi_{0} \prod_{i=1}^{k} f_{1}(Y_{i}^{1})}{\pi_{0} \prod_{i=1}^{k} f_{1}(Y_{i}^{1}) + (1-\pi_{0})\prod_{i=1}^{k} f_{0}(Y_{i}^{1})},
\end{eqnarray}
where $Y_{i}^{1}$ is the $i^{th}$ observation taken from the first sequence. In the reminder of this section, we replace $\eta_{1}$ by $\eta$ since they have the same distribution, and we replace $Y_{i}^{1}$ by $Y_{i}$ for the simplicity of the notation. Therefore, we have
\begin{eqnarray}
L_{k} := \prod_{i=1}^{k} L(Y_{i}) = \prod_{i=1}^{k} \frac{f_{1}(Y_{i})}{f_{0}(Y_{i})} = \frac{1-\pi_{0}}{\pi_{0}}\frac{\pi_{k}}{1-\pi_{k}},
\end{eqnarray}
where $L(Y_{i})=f_{1}(Y_{i})/f_{0}(Y_{i})$ is the LR of $Y_{i}$. Hence $\pi_{k}$ and $L_{k}$ have a one-to-one mapping. Let
\begin{eqnarray}
A = \frac{1-\pi_{0}}{\pi_{0}}\frac{\pi_{0}}{1-\pi_{0}}=1, \quad B = \frac{1-\pi_{0}}{\pi_{0}}\frac{\pi_{U}}{1-\pi_{U}}. \no
\end{eqnarray}
Denote $l(Y_{i}) = \log L(Y_{i})$ as the log likelihood ratio (LLR), and let us set
$$W_{k} := \log L_{k} = \sum_{i=1}^{k} l(Y_{i}).$$
Hence $W_{k}$ is a random walk. The stopping time $\eta$ can be equivalently written as
$$\eta = \inf\{k\geq0 | W_{k} \notin [\log A, \log B]\} = \inf\{k\geq0 | W_{k} \notin [0, \log B]\}.$$

Denote $\chi_{1}$ and $\chi_{0}$ as events $\{\pi_{\eta} > \pi_{U}\}$ and $\{\pi_{\eta} < \pi_{0} \}$, respectively. That is, $\chi_{1}$ is the event that the observer declares that the currently observing sequence is generated by $f_{1}$, and $\chi_{0}$ is the event that the observer switches to observe another sequence. Notice that $\chi_{0}$ can be equivalently written as $\{L_{k} < 1 \}$ or $\{W_{k} < 0 \}$, and $\chi_{1}$ can be equivalently written as $\{ L_{k} > B\}$ or $\{ W_{k} > \log B\}$. We also denote $\alpha := P_{0}(\chi_{1}) = 1-P_{0}(\chi_{0})$ and $\beta := P_{1}(\chi_{0})$ as Type I error and Type II error, respectively. According to \cite{Lai:TIT:11}, we have
\begin{eqnarray}
ASD &=& \frac{\mathbb{E}[\eta]}{\pi_{0}(1-\beta) + (1-\pi_{0})\alpha}\no\\
&=& \frac{\pi_{0}\mathbb{E}_{1}[\eta] + (1-\pi_{0})\mathbb{E}_{0}[\eta]}{\pi_{0}(1-\beta) + (1-\pi_{0})\alpha},\label{eq:single_delay} \\
FIP &=& \frac{(1-\pi_{0})\alpha}{\pi_{0}(1-\beta) + (1-\pi_{0})\alpha},\label{eq:single_error}
\end{eqnarray}
where $\mathbb{E}_{0}$ and $\mathbb{E}_{1}$ are expectations with respect to $P_{0}$ and $P_{1}$, respectively.

In the following, we study the performance of ASD when $H_{1}$ is rare. In this paper, the asymptotic analysis is in the sense $\pi_{0} \rightarrow 0$ rather than $\zeta \rightarrow 0$. In particular, we study two cases in the asymptotic analysis: $\zeta$ is constant within (0,1) and $\zeta \rightarrow 0$. The first case is referred to as the \emph{ fixed identification error} (FIE) case, and the second case is referred to as the \emph{rare identification error} (RIE) case. To proceed, we need to choose the value of threshold $B$ properly to satisfy $FIP \leq \zeta$.

\begin{lem}
For all $\zeta \in (0, 1)$, if $\pi_{0} \rightarrow 0$, then $\alpha \rightarrow 0$.
\end{lem}
\begin{proof}

By the constraint $FIP \leq \zeta$ and \eqref{eq:single_error}, it is easy to obtain that
\begin{eqnarray}
\alpha \leq (1 - \beta)\frac{\zeta}{1-\zeta}\frac{\pi_{0}}{1-\pi_{0}}. \label{eq:error2}
\end{eqnarray}
Hence, $\pi_{0} \rightarrow 0$ leads to $\alpha \rightarrow 0$.
\end{proof}

By the definition of $\alpha$, it is easy to see that $\alpha \rightarrow 0$ requires $B \rightarrow \infty$. Since $W_{k} = \sum_{i=1}^{k} l(Y_{i})$ is a random walk and $\mathbb{E}_{1}[l(Y_{i})] > 0$, the distribution of the overshoot that $W_{k}$ exceeds $\log B$ converges as $\log B \rightarrow \infty$ under $P_{1}$ (Theorem 8.25 in \cite{Siegmund:Book:85}). Let
\begin{eqnarray}
R(x) := \lim_{B \rightarrow \infty} P_{1}(W_{\eta}-\log B \leq x|W_{\eta} \geq \log B)
\end{eqnarray}
be the corresponding asymptotic cumulative distribution function (cdf).

\begin{thm} \label{thm:single_threshold}
As $B \rightarrow \infty$, we have
\begin{eqnarray}
\alpha = B^{-1}(1-\beta)\left(\int_{0}^{\infty} e^{-x}dR(x)\right).
\end{eqnarray}
\end{thm}
\begin{proof}
\begin{eqnarray}
\alpha &=& P_{0}(W_{\eta} \geq \log B) \no\\
&=& \sum_{k=1}^{\infty} P_{0}(W_{k} \geq \log B, \eta=k) \no\\
&=& \sum_{k=1}^{\infty} \int_{\{W_{k} \geq \log B, \eta=k\}} d P_{0} \no\\
&=& \sum_{k=1}^{\infty} \int_{\{W_{k} \geq \log B, \eta=k\}} \frac{d P_{0}}{d P_{1}} d P_{1} \no\\
&=& \sum_{k=1}^{\infty} \mathbb{E}_{1}[L_{k}^{-1}; W_{k} \geq \log B, \eta=k]  \no\\
&=& \mathbb{E}_{1}[L_{\eta}^{-1}; W_{\eta} \geq \log B] \no\\
&=& \mathbb{E}_{1}[L_{\eta}^{-1}| W_{\eta} \geq \log B] P_{1}( W_{\eta} \geq \log B). \label{eq:single_alpha}
\end{eqnarray}
Since $P_{1}( W_{\eta} \geq \log B) = 1 - \beta$, and
\begin{eqnarray}
\mathbb{E}_{1}[L_{\eta}^{-1}| W_{\eta} \geq \log B]
&=& \mathbb{E}_{1}[e^{-W_{\eta}}| W_{\eta} \geq \log B] \no\\
&=& e^{-\log B}\mathbb{E}_{1}[e^{-(W_{\eta}-\log B)}| W_{\eta} \geq \log B] \no\\
&=& B^{-1} \left(\int_{0}^{\infty} e^{-x}dR(x)\right), \no
\end{eqnarray}
the conclusion follows immediately.
\end{proof}

For the optimal solution, one needs to design $\pi_{U}^{*}$ (or $B^{*}$) such that \eqref{eq:error2} holds in equality. In general, $B^{*}$ is difficult to write explicitly due to the overshoot. However, for the asymptotic analysis, one can ignore the overshoot and find a simple threshold $B$ such that the algorithms with $B$ and $B^{*}$ achieve the same asymptotic delay.

\begin{cor} \label{cor:single_threshold} If $0< \beta < 1$ and
$$\int_{0}^{\infty} e^{-x}dR(x) < \infty, $$
then
$$B^{-1} = \frac{\zeta}{1-\zeta}\frac{\pi_{0}}{1-\pi_{0}}$$
is a threshold such that using $B$ and $B^{*}$, the algorithm has the same asymptotic behavior.
\end{cor}
\begin{proof}
By \eqref{eq:single_alpha}, we have
\begin{eqnarray}
\alpha = \mathbb{E}_{1}[L_{\eta}^{-1}| W_{\eta} \geq \log B] (1-\beta) < B^{-1} (1-\beta).
\end{eqnarray}
Hence, if we choose 
$$B^{-1} = \frac{\zeta}{1-\zeta}\frac{\pi_{0}}{1-\pi_{0}}$$
\eqref{eq:error2} is satisfied, which further indicate the FIP constraint is satisfied. Since threshold $B$ ignores only the effect of overshoot, $\int_{0}^{\infty} e^{-x}dR(x)$, which is a finite value, ASD has the same asymptotic behavior with $B$ and $B^{*}$.
\end{proof}

In the following, we analyze the asymptotic ASD under two cases: 1) the FIE case , i.e., $\pi_{0} \rightarrow 0$ while $\zeta$ is a constant in $(0, 1)$. By Corollary \ref{cor:single_threshold}, we have $|\log B| = |\log \pi_{0}|(1+o(1))$; 2) the RIE case, i.e., $\pi_{0} \rightarrow 0$ and $\zeta \rightarrow 0$. In this case, we have $|\log B| = (|\log \pi_{0}|+|\log \zeta|)(1+o(1))$. Let $\rho := \alpha/\pi_{0}$. By \eqref{eq:error2}, it is easy to see that $\rho$ is a constant in the FIE case and $\rho \rightarrow 0$ in the RIE case. We first have the following lemma:
\begin{lem} \label{lem:delaybound}
As $\pi_{0} \rightarrow 0$,
\begin{eqnarray}
&& 0 < \mathbb{E}_{0}[\eta] < \infty,   \no\\
&& 0 < \mathbb{E}_{1}[\eta] \leq (1-\beta)\frac{|\log B|}{D(f_{1}||f_{0})}(1+o(1)), \no
\end{eqnarray}
in which $D(f_{1}||f_{0})$ is the KL divergence of $f_{1}$ and $f_{0}$.
\end{lem}
\begin{proof}
It is obvious that $\mathbb{E}_{0}[\eta] > 0$ and $\mathbb{E}_{1}[\eta] > 0$.
The proof of $\mathbb{E}_{0}[\eta] < \infty $ follows exactly the proof of Lemma 1 in \cite{Banerjee:TIT:12}. Hence, we only need to show the upper bound of $\mathbb{E}_{1}[\eta]$. Since $W_{k} = \sum_{i=1}^{k} l(Y_{i})$, and $Y_{i}$'s are i.i.d. for $1\leq i \leq \eta$. By Wald's identity, we have
\begin{eqnarray}
\mathbb{E}_{1}[W_{\eta}] = \mathbb{E}_{1}[\eta]\mathbb{E}_{1}[l(Y_{1})]. \no
\end{eqnarray}
At the same time, we have
\begin{eqnarray}
\mathbb{E}_{1}[W_{\eta}] &=& \mathbb{E}_{1}[W_{\eta}|W_{\eta}<0]P_{1}(W_{\eta}<0) \no\\
&+& \mathbb{E}_{1}[W_{\eta}|W_{\eta}> \log B]P_{1}(W_{\eta} > \log B) \no\\
&\leq& \mathbb{E}_{1}[W_{\eta}|W_{\eta}> \log B]P_{1}(W_{\eta} > \log B) \no\\
&=& |\log B| P_{1}(\chi_{1})(1+o(1)) \no\\
&=& (1-\beta)|\log B|(1+o(1)). \no
\end{eqnarray}
As the result, we have
\begin{eqnarray}
\mathbb{E}_{1}[\eta] = \frac{\mathbb{E}_{1}[W_{\eta}]}{\mathbb{E}_{1}[l(Y_{1})]} \leq (1-\beta)\frac{|\log B|}{D(f_{1}||f_{0})}(1+o(1)). \no
\end{eqnarray}
\end{proof}

\begin{thm}
If $0 < D(f_{1}||f_{0})<\infty$, then as $\pi_{0} \rightarrow 0$, ASD for the FIE case is given as
\begin{eqnarray}
ASD = \frac{1-\pi_{0}}{\rho(1-\pi_{0})+(1-\beta)} \frac{1}{\pi_{0}} \mathbb{E}_{0}[\eta] (1+o(1)).
\end{eqnarray}
In addition, if $\pi_{0}|\log \zeta| \rightarrow 0$, then ASD for the RIE case is given as
\begin{eqnarray}
ASD = \frac{1-\pi_{0}}{1-\beta} \frac{1}{\pi_{0}} \mathbb{E}_{0}[\eta] (1+o(1)).
\end{eqnarray}
\end{thm}
\begin{proof}
By \eqref{eq:single_delay}, ASD can be written as
\begin{eqnarray}
ASD &=& \frac{1}{\rho(1-\pi_{0})+(1-\beta)}\frac{1-\pi_{0}}{\pi_{0}}\left(\frac{\pi_{0}}{1-\pi_{0}}\mathbb{E}_{1}[\eta]+\mathbb{E}_{0}[\eta]\right). \no
\end{eqnarray}
From the discussion after Corollary \ref{cor:single_threshold}, we know that $|\log B| = |\log \pi_{0}|(1+o(1))$ in the FIE case and $|\log B| = (|\log \pi_{0}|+|\log \zeta|)(1+o(1))$ in the RIE case. It is easy to verify that $\pi_{0}|\log B| \rightarrow 0$ under both of these cases. Therefore
\begin{eqnarray}
\frac{\pi_{0}}{1-\pi_{0}}\mathbb{E}_{1}[\eta] \leq (1-\beta) \frac{\pi_{0}}{1-\pi_{0}} \frac{|\log B|}{D(f_{1}||f_{0})}(1+o(1)) \rightarrow 0. \no
\end{eqnarray}
Then, the first conclusion follows immediately; the second conclusion can be obtained by noticing the fact that $\rho\rightarrow 0$ in the RIE case.
\end{proof}

\begin{rmk}
In the above theorem, we introduce an additional condition that $\pi_{0}|\log \zeta| \rightarrow 0$ to limit the speed of $\zeta$ approaching zero for the RIE case. This condition could be easily satisfied. For example, when $\zeta$ goes to zero on the order $\pi_{0}^{n}$ for any $n<\infty$, this condition still holds.
\end{rmk}

\section{New Search Strategy Based on Mixed Observations} \label{sec:mixed}
\subsection{New Strategy} \label{subsec:mixed_model}
In this section we propose a new search strategy, termed \emph{mixed observation} search strategy, for the multi-sequence search problem described in Section \ref{subsec:single_model}. The mixed observation search is a sequential strategy consisting of two stages, namely scanning and refinement stages respectively.

In the scanning stage, the observer picks two sequences $s_k^{1}$ and $s_k^{2}$ at each time slot $k$ and observes a linear combination of samples from these two sequences:
\begin{eqnarray}
Z_k=a_1Y_{k}^{s_k^{1}}+a_2Y_{k}^{s_k^{2}}.
\end{eqnarray}
Since $s_k^{1}$ and $s_k^{2}$ have no difference in their distribution, we simply set the same weight $a_1=a_2=1$ in our search strategy. We note that this choice may not be optimal since the coefficients $a_{1}$ and $a_{2}$ can be updated at every time slot based on the previous observations. Hence, to optimize the linear combination is one of our future research directions. In this paper, we focus on the setting $a_1=a_2=1$ and we show that even this simple setting can bring significant improvement when the occurrence of $f_{1}$ is rare.

Since each sequence has two possible pdfs, $Z_k$ has three possible pdfs: 1) $g_{0} := f_0*f_0$, which happens when both sequences $s_k^{1}$ and $s_k^{2}$ are generated from $f_0$. Here $*$ denotes the convolution. The prior probability of this occurring is $p^{0,0}_0=(1-\pi_{0})^2$; 2) $g_{1} := f_0*f_1$, which happens when one of these two sequences is generated from $f_0$ and the other one is generated from $f_1$. The prior probability of this occurring is $p^{mix}_0=2\pi_{0}(1-\pi_{0})$; and 3) $g_{2} := f_1*f_1$, which happens when both sequences are generated from $f_1$. The prior probability of this occurring is $p^{1,1}_0=\pi_{0}^2$. Here, we use $g$ to represent the pdf of $Z_{k}$, and use the subscript of $g$ to denote the number of sequences generated from $f_{1}$.

With a little abuse of notation, we use $\{\mathcal{F}_k\}$ to denote the filtration generated from the observations in the scanning stage, i.e., $\mathcal{F}_k=\sigma\{Z_1,\cdots,Z_k\}$. After taking sample $Z_k$, the observer needs to make the following two decisions: 1) whether to stop the scanning stage and enter the refinement stage to further examine the two sequences more closely. We use $\tau_0$ to denote the time that the observer stops the scanning stage. Hence $\tau_{0}$ is a stopping time with respect to $\{\mathcal{F}_{k}\}$; and 2) if the decision is to continue the scanning process, the observer needs to decide whether to take more samples from the same two sequences, or to switch to two new sequences. We still use $\phi_{k}(\mathcal{F}_k)$ to denote the switch function. If $\phi_{k}(\mathcal{F}_k)=1$, the observer switches to observe a pair of new sequences, while if $\phi_{k}(\mathcal{F}_k)=0$, the observer takes more samples from the currently observing sequences. Let $\phiv= \{ \phi_1,\phi_2,\cdots \}$ be the sequence of switch decisions. Same as \cite{Lai:TIT:11}, we assume that if sequences have been abandoned, the observer will not come back and exam them again.

In this proposed strategy, we emphasize that once the observer enters the refinement stage, it could not come back to the scanning stage any more. Hence the observer will not enter the refinement stage until he is confident that at least one of the observing sequence is generated from $f_{1}$. The extension to case in which the observer can reenter the scanning stage will be a subject of future study.

In the refinement stage, the observer examines the two candidate sequences more closely. Each sample taken during the refinement stage will come from one sequence. Hence, at this stage, no mixing is used anymore. We will use $j$ to denote the index of samples taken at this stage. Clearly, at the beginning of the refinement stage, i.e. $j=1$, there is no difference between these two candidates $s_{\tau_0}^{1}$ and $s_{\tau_0}^{2}$, and hence the observer simply picks one $s_{\tau_0}^{1}$:
\begin{eqnarray}
X_j=Y_{\tau_0+j}^{s_{\tau_0}^{1}}.
\end{eqnarray}
After taking each sample, the observer needs to decide whether or not to stop the refinement stage. If the observer decides to stop, the observer should choose one of the two candidate sequences and claim that it is generated from $f_1$. Intuitively, if the observer believes that the observed candidate sequence $s_{\tau_0}^{1}$ is generated from $f_{1}$, then the observer claims $s_{\tau_0}^{1}$. Otherwise, the observer claims $s_{\tau_0}^{2}$. Hence, at the end of the refinement stage, one of the candidate sequences must be declared to be generated from $f_{1}$. Let $\mathcal{G}_j=\sigma(Z_1,\cdots,Z_{\tau_0},X_1,\cdots,X_j)$ be the filtration generated by the observations from these two stages. We use $\tau_1$ to denote the time at which the observer stops the refinement stage, hence, $\tau_{1}$ is a stopping time with respect to $\{ \mathcal{G}_{j} \}$. Let $\delta$ be the terminal decision rule, according to which the observer picks the sequence that is claimed be generated from $f_1$.

We are still interested in ASD and FIP. In this case, ASD is defined as
$$ASD_{m} = \mathbb{E}[\tau_0+\tau_1],$$
and FIP is defined as
$$FIP_{m} = P(H^{\delta}=H_0).$$
We want to solve the following optimization problem
\begin{eqnarray}
\inf\limits_{\tau_0,\phiv,\tau_1,\delta} ASD_{m} \text{  subject to  } FIP_{m} \leq \zeta. \label{eq:mixed_P1}
\end{eqnarray}
For any given $\zeta$, by Lagrange multiplier this problem can be equivalently written as
\begin{eqnarray}\label{eq:cost}
\inf\limits_{\tau_0,\phiv,\tau_1,\delta} c\mathbb{E}[\tau_0+\tau_1]+P(H^{\delta}=H_0) \label{eq:mixed_P2}
\end{eqnarray}
for a properly chosen constant $c$.

\subsection{Optimal Solution} \label{sec:mixed_opt}
In this subsection, we discuss the optimal solution for the proposed mixed observation search strategy. We first introduce some important statistics used in the optimal solution.

For the scanning stage, after taking $k$ observations, we define the following posterior probabilities:
\begin{eqnarray}
&&p^{1,1}_k := P\left( \text{both } s_k^1 \text{ and } s_{k}^2 \text{ are generated from}  f_1 | \mathcal{F}_k \right), \no \\
&&p^{mix}_k := P\left( \text{one of } s_k^1 \text{ and } s_{k}^2 \text{ is generated from}  f_1| \mathcal{F}_k \right), \no \\
&&p^{0,0}_k := P\left( \text{both } s_k^1 \text{ and } s_{k}^2 \text{ are generated from}  f_0| \mathcal{F}_k \right). \no
\end{eqnarray}
As discussed in the previous subsection, at the beginning of the scanning stage we have $p^{1,1}_0 = \pi_{0}^2$, $p^{mix}_0 = 2\pi_{0}(1-\pi_{0})$ and $p^{0,0}_0 = (1-\pi_{0})^2$. Let $\mathbf{p}_{k}=[p^{1,1}_{k}, p^{mix}_{k}, p^{0,0}_{k}]$. It is easy to check that these posterior probabilities can be updated as follows:
\begin{eqnarray}
&&p^{1,1}_{k+1}=\frac{p^{1,1}_{k}g_{2}(Z_{k+1})}{g(\mathbf{p}_{k}, Z_{k+1})}\mathbf{1}_{\{\phi_k=0\}}+\frac{p^{1,1}_{0}g_{2}(Z_{k+1})}{g(\mathbf{p}_{0}, Z_{k+1})}\mathbf{1}_{\{\phi_k=1\}},\no\\
&&p^{mix}_{k+1}=\frac{p^{mix}_{k}g_{1}(Z_{k+1})}{g(\mathbf{p}_{k}, Z_{k+1})}\mathbf{1}_{\{\phi_k=0\}}+\frac{p^{mix}_{0}g_{1}(Z_{k+1})}{g(\mathbf{p}_{0}, Z_{k+1})}\mathbf{1}_{\{\phi_k=1\}},\no\\
&&p^{0,0}_{k+1}=1-p^{1,1}_{k+1}-p^{mix}_{k+1}, \no
\end{eqnarray}
where $g(\mathbf{p}_{k}, z_{k+1})$ and $g(\mathbf{p}_{0}, z_{k+1})$ are defined as
\begin{eqnarray}
&&g(\mathbf{p}_{k}, z_{k+1}) := p^{0,0}_{k}g_{0}(z_{k+1})+p^{mix}_kg_{1}(z_{k+1})+p^{1,1}_kg_{2}(z_{k+1}),\no\\
&&g(\mathbf{p}_{0}, z_{k+1}) := p^{0,0}_{0}g_{0}(z_{k+1})+p^{mix}_0g_{1}(z_{k+1})+p^{1,1}_0g_{2}(z_{k+1}).\no
\end{eqnarray}
Hence $\mathbf{p}_{k}$ satisfies the Markov property.

For the refinement stage, after taking $j$ observations, we define
\begin{eqnarray}
&&r^{1,1}_{j} := P\left(\text{both } s_{\tau_0}^1 \text{ and } s_{\tau_0}^2 \text{ are generated from}  f_1 |\mathcal{G}_j \right), \no \\
&&r^{1,0}_{j} := P\left(s_{\tau_0}^1 \text{ is generated from } f_1, s_{\tau_0}^2 \text{ is generated from} f_0 |\mathcal{G}_j \right), \no \\
&&r^{0,1}_{j} := P\left(s_{\tau_0}^1 \text{ is generated from } f_0, s_{\tau_0}^2 \text{ is generated from} f_1|\mathcal{G}_j \right), \no \\
&&r^{0,0}_{j} := P\left(\text{both } s_{\tau_0}^1 \text{ and } s_{\tau_0}^2 \text{ are generated from}  f_0|\mathcal{G}_j \right). \no
\end{eqnarray}
At the beginning of the refinement stage, we have $r^{1,1}_{0} = p_{\tau_{0}}^{1,1}$ and $r^{1,0}_{0} = r^{0,1}_{0} = p_{\tau_{0}}^{mix}/2$. It is easy to verify that these statistics can be updated using
\begin{eqnarray}
&&\hspace{-6mm} r^{1,1}_{j+1} = \frac{f_{1}(X_{j+1}) r^{1,1}_{j}}{f_{1}(X_{j+1})(r^{1,1}_{j} + r^{1,0}_{j}) + f_{0}(X_{j+1}) (r^{0,1}_{j}+r^{0,0}_{j})}, \no\\
&&\hspace{-6mm} r^{1,0}_{j+1} = \frac{f_{1}(X_{j+1}) r^{1,0}_{j}}{f_{1}(X_{j+1})(r^{1,1}_{j} + r^{1,0}_{j}) + f_{0}(X_{j+1}) (r^{0,1}_{j}+r^{0,0}_{j})}, \no\\
&&\hspace{-6mm} r^{0,1}_{j+1} = \frac{f_{0}(X_{j+1}) r^{0,1}_{j}}{f_{1}(X_{j+1})(r^{1,1}_{j} + r^{1,0}_{j}) + f_{0}(X_{j+1}) (r^{0,1}_{j}+r^{0,0}_{j})}, \no\\
&&\hspace{-6mm} r^{0,0}_{j+1} = 1 - r^{1,1}_{j+1} - r^{1,0}_{j+1} - r^{0,1}_{j+1}. \no
\end{eqnarray}
Let $\mathbf{r}_{j} = [r^{1,1}_{j}, r^{1,0}_{j}, r^{0,1}_{j}, r^{0,0}_{j}]$, hence  $\mathbf{r}_{j}$ satisfies the Markov property. For the brevity of notation, we further define the following two statistics
\begin{eqnarray}
q_{1, j} &:=& r^{1,1}_{j} + r^{1,0}_{j}\no\\
&=&P\left(s_{\tau_0}^1 \text{ is generated from } f_1 | \mathcal{G}_j \right), \label{eq:q1}\\
q_{2, j} &:=& r^{1,1}_{j} + r^{0,1}_{j}\no\\
&=&P\left(s_{\tau_0}^2 \text{ is generated from } f_1 | \mathcal{G}_j \right). \label{eq:q2}
\end{eqnarray}

Using the above defined statistics, we first have the following theorem about the optimal terminal decision rule:
\begin{thm} \label{thm:error_prob}
For any $\tau_0,\phiv$ and $\tau_1$, the optimal terminal decision rule is given as
\begin{eqnarray}
\delta^{*}=\left\{\begin{array}{ll}s_{\tau_0}^1 & \text{ if } q_{1, \tau_1} > q_{2, \tau_1} \\
s_{\tau_0}^2 &\text{ if } q_{1, \tau_1} \leq q_{2, \tau_1} \end{array}\right. ,
\end{eqnarray}
and the corresponding cost is given as
\begin{eqnarray}
\inf\limits_{\delta} P\left(H^{\delta}=H_0\right) = \mathbb{E}\left[1-\max\left\{q_{1, \tau_1}, q_{2, \tau_1}\right\}\right].
\end{eqnarray}
\end{thm}
\begin{proof}
Please see Appendix \ref{app:error_prob}.
\end{proof}
Hence, the optimal terminal rule is to pick the sequence with the larger posterior probability. Moreover, this theorem converts the cost of FIP into a function of $q_{1, j}$ and $q_{2, j}$, which is a function of the refinement stage statistics $\mathbf{r}_{j}$. Similar to the reduction method proposed in \cite{Kobylanski:AAP:11} (Theorem 2.3 in \cite{Kobylanski:AAP:11}), \eqref{eq:mixed_P2} can be decomposed into two concatenated single stopping time problems. In particular, we first solve the optimal stopping time $\tau_{1}$ for any given $\tau_{0}$ and $\phiv$, then with the optimal $\tau_{1}$, we solve the other stopping time $\tau_{0}$ with corresponding $\phiv$. The decomposition is stated in the following lemma:
\begin{lem} \label{lem:multi_stopping}
Let
\begin{eqnarray}
w_{0} &:=& \inf_{\tau_{0}, \phiv, \tau_{1}} \mathbb{E}\left[c(\tau_{0}+\tau_{1})+1-\max\left\{ q_{1, \tau_{1}}, q_{2, \tau_{1}} \right\}\right], \label{eq:v0}\\
v(\tau_{0}, \phiv) &:=& \inf_{\tau_{1}} \mathbb{E}\left[c\tau_{1}+1-\max\left\{ q_{1, \tau_{1}}, q_{2, \tau_{1}}\right\}| \mathcal{F}_{\tau_{0}}\right], \label{eq:w0}\\
u_{0} &:=& \inf_{\tau_{0}, \phiv} \mathbb{E}\left[c\tau_{0}+v(\tau_{0}, \phiv)\right]. \label{eq:u0}
\end{eqnarray}
Then
$$w_{0} = u_{0}.$$
\end{lem}
\begin{proof}
Please see Appendix \ref{app:multi_stopping}.
\end{proof}

The above lemma converts the original problem $w_{0}$ into two single optimal stopping problems $v(\tau_{0}, \phiv)$ and $u_{0}$, which correspond to the cost functions in the refinement stage and the scanning stage, respectively. We can tackle these two problems one by one. The optimal stopping rule for the refinement stage is given as:
\begin{thm} \label{thm:refine}
For any given $\tau_0$ and $\phiv$,
\begin{eqnarray}
v(\tau_{0}, \phiv) &=& V(r^{1,1}_{0}, r^{1,0}_{0}, r^{0,1}_{0}), \no
\end{eqnarray}
in which $V(\cdot)$ is a function that satisfies the following recursion:
\begin{eqnarray}
V(\mathbf{r}_{j}) &=& \min \left\{ 1 - \max\left\{q_{1, j}, q_{2, j} \right\}, c + \mathbb{E}\left[ V(\mathbf{r}_{j+1}) | \mathbf{r}_{j} \right] \right\}. \no
\end{eqnarray}
In addition, the optimal stopping time $\tau_{1}$ for \eqref{eq:w0} is given as
\begin{eqnarray}\
\tau_{1}^{*} = \inf\left\{ j \geq 0: 1 - \max\left\{q_{1, j}, q_{2, j} \right\} \leq c+ \mathbb{E}\left[ V(\mathbf{r}_{j+1}) | \mathbf{r}_{j} \right] \right\}. \no
\end{eqnarray}
\end{thm}
\begin{proof}
Please see Appendix \ref{app:refine}.
\end{proof}
We note that the form of $V(\mathbf{r}_{j})$ can be obtained via an iterative procedure offline~\cite{Poor:Book:08}. This theorem indicates that the optimal strategies in the refinement stage are related to $\tau_0, \phiv$ only through $p^{1,1}_{\tau_{0}},p^{mix}_{\tau_{0}}$, since
\begin{eqnarray}
v(\tau_{0}, \phiv) &=& V(r^{1,1}_{0}, r^{1,0}_{0}, r^{0,1}_{0}) \no\\
&=& V(p^{1,1}_{\tau_{0}}, p^{mix}_{\tau_{0}}/2, p^{mix}_{\tau_{0}}/2) =: v(p^{1,1}_{\tau_{0}}, p^{mix}_{\tau_{0}}). \no
\end{eqnarray}
Hence, we denote $v(\tau_{0}, \phiv)$ as $v\left(p^{1,1}_{\tau_{0}},p^{mix}_{\tau_{0}}\right)$ in the following discussion.
$v\left(p^{1,1}_{\tau_{0}},p^{mix}_{\tau_{0}}\right)$ is defined over the domain
\begin{eqnarray}
\mathcal{P} = \left\{ \left(p^{1,1},p^{mix}\right): 0 \leq p^{1,1} \leq 1, 0 \leq p^{mix} \leq 1, 0 \leq p^{1,1}+p^{mix} \leq 1 \right\}.\no
\end{eqnarray}

\begin{lem} \label{lem:refine}
$v\left(p^{1,1},p^{mix}\right)$ is a concave function over $\mathcal{P}$
with $v(1,0) = 0$ and $v(0,0) = 1$.
\end{lem}
\begin{proof}
Please see Appendix \ref{app:refine}.
\end{proof}

As the result, \eqref{eq:u0} can be written as
\begin{eqnarray}
u_{0} = \inf\limits_{\tau_0,\phiv}\mathbb{E}\left[c\tau_0 + v\left(p^{1,1}_{\tau_{0}},p^{mix}_{\tau_{0}}\right)\right], \no
\end{eqnarray}
and its optimal solution is given as
\begin{thm} \label{thm:scan1}
The optimal stopping rule for the scanning stage is given as
\begin{eqnarray}
\tau_{0}^{*} = \inf\left\{ k\geq 0: v\left(p_{k}^{1,1}, p_{k}^{mix}\right) = U\left(p_{k}^{1, 1}, p_{k}^{mix}\right) \right\}
\end{eqnarray}
and the optimal switching rule is given as
\begin{eqnarray}
\phi_k^{*} = \left\{\begin{array}{ll} 0 & \text{ if } \Phi_{c}\left(p_{k}^{1,1}, p_{k}^{mix}\right) \leq \Phi_{s}\\
1 &\text{ otherwise } \end{array}\right.,
\end{eqnarray}
in which, $U(\cdot)$ is a function that satisfies the following operator
\begin{eqnarray}
U\left(p_{k}^{1, 1}, p_{k}^{mix}\right) = \min \left\{ v\left(p^{1,1}, p_{k}^{mix}\right),  c + \min \left\{ \Phi_{c}\left(p_{k}^{1, 1}, p_{k}^{mix}\right), \Phi_{s} \right\} \right\} \no
\end{eqnarray}
with
\begin{eqnarray}
&&\hspace{-7mm} \Phi_{c}\left(p_{k}^{1, 1}, p_{k}^{mix}\right) = \mathbb{E}\left[U\left(p_{k+1}^{1, 1}, p_{k+1}^{mix}\right)\Big|p_{k}^{1, 1}, p_{k}^{mix}, \phi_{k}=0 \right], \no\\
&&\hspace{-7mm} \Phi_{s} = \mathbb{E}\left[U\left(p_{k+1}^{1, 1}, p_{k+1}^{mix}\right)\Big|p_{k}^{1, 1}, p_{k}^{mix}, \phi_{k}=1 \right]. \no
\end{eqnarray}
\end{thm}
\begin{proof}
Please see Appendix \ref{app:scan}.
\end{proof}

\begin{rmk}
One can show that $\Phi_{s} = \Phi_{c}\left(p_{0}^{1, 1}, p_{0}^{mix}\right)$, hence it is a constant between 0 and 1. For this reason, we denote it as $\Phi_{s}$ rather than $\Phi_{s}\left(p_{k}^{1, 1}, p_{k}^{mix}\right)$ in the above theorem.
\end{rmk}

Same as the refinement stage, all the functions involved in the above theorem can be computed offline. The optimal solutions of $\tau_{0}^{*}$ and $\phi_{k}^{*}$ can be further simplified using the following lemma.
\begin{lem} \label{lem:scan}
1) $\Phi_{c}\left(p^{1,1}, p^{mix}\right)$ is a concave function over $\mathcal{P}$, and $0 \leq \Phi_{c}\left(p^{1, 1}, p^{mix}\right) \leq 1$. \\ 
2) $U\left(p^{1, 1}, p^{mix}\right)$ is a concave function over domain $\mathcal{P}$, and $0 \leq U\left(p^{1, 1}, p^{mix}\right) \leq 1$.  
\end{lem}
\begin{proof}
Please see Appendix \ref{app:scan}.
\end{proof}

Since both $U\left(p^{1, 1}, p^{mix}\right)$ and $v\left(p^{1, 1}, p^{mix}\right)$ are concave functions over $\mathcal{P}$, $U\left(p^{1, 1}, p^{mix}\right) \leq v\left(p^{1, 1}, p^{mix}\right)$ over $\mathcal{P}$, and $U(1, 0) = v(1,0) = 0$, there must exist a certain region, denoted as $R_{\tau}$, on which these two concave surfaces are equal to each other. Hence, the optimal stopping time $\tau_{0}^{*}$ can be described as the first hitting time of the process $\left(p_{k}^{1, 1}, p_{k}^{mix}\right)$ to the region $R_{\tau}$. Similarly, $\Phi_{c}$ is a concave surface and $\Phi_{s}$ is a constant plane with $\Phi_{s} = \Phi_{c}(p^{1,1}_{0}, p^{mix}_{0})$. Hence, $\mathcal{P}$ can be divided into two connected regions $R_{\phi}$ and $\mathcal{P}\backslash R_{\phi}$, where $R_{\phi} := \left\{\left(p^{1,1}, p^{mix}\right): \Phi_{c}(p^{1,1}, p^{mix}) \leq \Phi_{s} \right\}$. Hence, the observer switches to new sequences at time slot $k$ if $\left(p^{1,1}_{k}, p^{mix}_{k}\right)$ is in $R_{\phi}$. We illustrate these two regions in Figure \ref{fig:region1}. As the result, we have the following theorem.

\begin{thm}\label{thm:scan2}
There exist two regions, $R_{\tau} \subset \mathcal{P}$ and $R_{\phi} \subset \mathcal{P} $, such that
\begin{eqnarray}
\tau_0^{*} = \min\left\{k \geq 0: \left(p_{k}^{1,1}, p_{k}^{mix}\right) \in R_{\tau} \right\},
\end{eqnarray}
and
\begin{eqnarray}
\phi_k^{*} =\left\{\begin{array}{ll} 1 & \text{ if } (p_{k}^{1,1}, p_{k}^{mix}) \in R_{\phi}\\
0 &\text{ otherwise } \end{array}\right..
\end{eqnarray}
\end{thm}

\begin{figure}[thb]
\centering
\includegraphics[width=0.35 \textwidth]{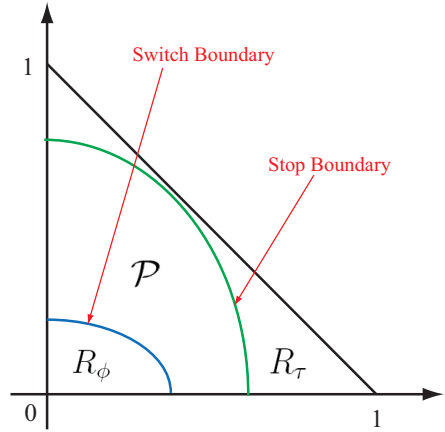}
\caption{An illustration of the region rule}
\label{fig:region1}
\end{figure}

\section{A Low Complexity Algorithm for the Mixed Search Strategy} \label{sec:mixed_asym}
As we can see from the previous section, the optimal solution of the mixed search strategy has a very complex structure. In this section, we propose a simple detection strategy and analyze its performance when $H_{1}$ is rare.

\subsection{A Low Complexity Algorithm} \label{sec:mixed_alg}
We propose a low complexity search algorithm in which the observer adopts the CUSUM test in the scanning stage and adopts SPRT in the refinement stage. Specifically, in the scanning stage, we use
\begin{eqnarray}
&& \tau_{0} = \inf\{k>0 : \tilde{p}_k < p_{L}\}, \no\\
&& \phi_{k} = \left\{ \begin{array}{cc}
           1 & \text{ if } \tilde{p}_k > p_{U} \\
           0 & \text{ otherwise }
         \end{array} \right.,  \label{eq:scan_CUSUM}
\end{eqnarray}
where $\tilde{p}_k$ is computed recursively using the following formula
\begin{eqnarray}
\tilde{p}_{k+1} &=& \frac{\tilde{p}_{k} g_{0}(Z_{k+1})}{\tilde{p}_{k} g_{0}(Z_{k+1}) + (1-\tilde{p}_{k})g_{1}(Z_{k+1})}\mathbf{1}_{\{\phi_{k}=0\}} \no\\
&+& \frac{p_{0}^{0, 0} g_{0}(Z_{k+1})}{p_{0}^{0, 0} g_{0}(Z_{k+1}) + (1-p_{0}^{0, 0})g_{1}(Z_{k+1})}\mathbf{1}_{\{\phi_{k}=1\}}. \no
\end{eqnarray}
In the refinement stage, we use
\begin{eqnarray}
&& \tau_{1} = \inf\left\{j>0 : q_{1, j} \notin [q_{L}, q_{U}]\right\}, \no\\
&& \delta = \left\{ \begin{array}{cc}
           s_{\tau_{0}}^{1} & \text{ if } q_{1, j} > q_{U} \\
           s_{\tau_{0}}^{2} & \text{ if } q_{1, j} < q_{L}
         \end{array} \right.. \label{eq:refine_SPRT}
\end{eqnarray}
Here $q_{1,j}$ is defined in \eqref{eq:q1}, and $p_{L}$, $p_{U}$, $q_{L}$, $q_{U}$ are pre-designed constant thresholds. The selection of these four thresholds will be discussed in the sequel.
\begin{rmk}
From Section \ref{sec:mixed_opt}, we know that $q_{1,0}$ is the sum of $r_{0}^{1,1}$ and $r_{0}^{1,0}$, which are determined by $p_{\tau_{0}}^{mix}$ and $p_{\tau_{0}}^{1,1}$. However, in the scanning stage, the proposed algorithm does not contain these two statistics. For the implementation purpose, we can simply set $q_{1,0} = 1/2$. This choice will be justified in Lemma \ref{lem:refine_delay}, which shows that the initial value of $q_{1,0}$ does not significantly affect the total search delay. In the following derivations, we keep using the notation $q_{1,0}$.
\end{rmk}

\begin{rmk}
The above proposed strategy can be expressed equivalently in terms of likelihood ratios. For the refinement stage, since $q_{1,j} = r_{j}^{1,1}+r_{j}^{1,0}$, for $j=0, 1, 2, \ldots$, we have
\begin{eqnarray}
q_{1,j+1} &=& \frac{f_{1}(X_{j+1}) (r^{1,1}_{j}+r^{1,0}_{j})}{f_{1}(X_{j+1})(r^{1,1}_{j} + r^{1,0}_{j}) + f_{0}(X_{j+1}) (r^{0,1}_{j}+r^{0,0}_{j})} \no\\
&=& \frac{f_{1}(X_{j+1}) q_{1,j} }{f_{1}(X_{j+1})q_{1,j} + f_{0}(X_{j+1}) (1-q_{1,j})} \no\\
&=& \frac{ q_{1,0} \prod_{i=1}^{j+1} f_{1}(X_{i}) }{ q_{1,0} \prod_{i=1}^{j+1} f_{1}(X_{i}) + (1-q_{1,0}) \prod_{i=1}^{j+1} f_{0}(X_{i})}. \no
\end{eqnarray}
Hence we have
\begin{eqnarray}
L^{(refine)}_{j} := \prod_{i=1}^{j}\frac{f_{1}(X_{i})}{f_{0}(X_{i})} = \frac{q_{1,j}}{1-q_{1,j}} \frac{1-q_{1,0}}{q_{1,0}}.
\end{eqnarray}
That is, $q_{1,j}$ and $L^{(refine)}_{j}$ have a one-to-one mapping for any given $q_{1,0}$. For the scanning stage, let us define the following stopping time
$$\eta_{m} = \inf\left\{k>0: \tilde{p}_{k} \notin [p_{L}, p_{U}]\right\}.$$
Similar to the single observation strategy, the stopping time $\tau_{0}$ in the scanning stage can also be viewed as a renewal process, with each renewal occurring whenever $\tilde{p}_{k}$ is reset to $p^{0,0}_{0}$, and with a termination occurring whenever $\tilde{p}_{k}$ exits the lower bound $p_{L}$, hence we denote
$$\tau_{0} = \sum_{n=1}^{N} \eta_{m, n},$$
where $\eta_{m,1}, \ldots, \eta_{m,n}, \ldots$ are i.i.d. repetition of $\eta_{m}$, and $N$ is the number of repetition. Notice that from $\eta_{m,n}+1$ to $\eta_{m, n+1}$, the observer is observing the same sequence. Since $\eta_{m,n}$'s are i.i.d, we can only focus on the the evolution of $\tilde{p}_{k}$ for $k = 1, 2, \ldots, \eta_{m}$, we have
\begin{eqnarray}
\tilde{p}_{k} &=& \frac{\tilde{p}_{k-1}g_{0}(Z_{k})}{\tilde{p}_{k-1}g_{0}(Z_{k}) + (1-\tilde{p}_{k-1})g_{1}(Z_{k})} \no\\
&=& \frac{p_{0}^{0, 0} \prod_{i=1}^{k}g_{0}(Z_{i})}{p_{0}^{0, 0} \prod_{i=1}^{k}g_{0}(Z_{i}) + (1-p_{0}^{0,0}) \prod_{i=1}^{k}g_{1}(Z_{i})}. \no
\end{eqnarray}
Hence,
\begin{eqnarray}
L^{(scan)}_{k} := \prod_{i=1}^{k} \frac{g_{1}(Z_{i})}{g_{0}(Z_{i})} = \frac{p_{0}^{0,0}}{1-p_{0}^{0,0}} \frac{1-\tilde{p}_{k}}{\tilde{p}_{k}}.
\end{eqnarray}
Therefore, the proposed strategy in  \eqref{eq:scan_CUSUM} and \eqref{eq:refine_SPRT} can be equivalently written as
\begin{eqnarray}
&& \tau_{0} = \inf\left\{ k>0 : L^{(scan)}_{k} > A_{s} \right\}, \no\\
&& \phi_{k} = \left\{ \begin{array}{cc}
           1 & \text{ if } L^{(scan)}_{k} < B_{s} \\
           0 & \text{ otherwise }
         \end{array} \right.,  \label{eq:scan_CUSUM2} \\
&& \tau_{1} = \inf\left\{j>0 : L^{(refine)}_{j} \notin [A_{r}, B_{r}]\right\}, \no\\
&& \delta = \left\{ \begin{array}{cc}
           s_{\tau_{0}}^{1} & \text{ if } L^{(refine)}_{\tau_{1}} > B_{r} \\
           s_{\tau_{0}}^{2} & \text{ if } L^{(refine)}_{\tau_{1}} < A_{r}
         \end{array} \right., \label{eq:refine_SPRT2}
\end{eqnarray}
in which the thresholds are given as
\begin{eqnarray}
A_{s} = \frac{p_{0}^{0,0}}{1-p_{0}^{0,0}}\frac{1-p_{U}}{p_{U}}, \quad B_{s} = \frac{p_{0}^{0,0}}{1-p_{0}^{0,0}}\frac{1-p_{L}}{p_{L}}, \no\\
A_{r} = \frac{1-q_{1,0}}{q_{1,0}}\frac{q_{L}}{1-q_{L}}, \quad  B_{r} = \frac{1-q_{1,0}}{q_{1,0}}\frac{q_{U}}{1-q_{U}}. \no
\end{eqnarray}
\end{rmk}

\begin{rmk}
Based on the above discussion, an intuitive explanation of the proposed strategy is given as follows: in the scanning stage, the observer is supposed to do the following trinary simple hypothesis test at each time slot:
\begin{eqnarray}
&&H_{0,0}: Z_{k} \sim P_{0,0}, \no\\
&&H_{mix}: Z_{k} \sim P_{mix}, \no\\
&&H_{1,1}: Z_{k} \sim P_{1,1}, \no
\end{eqnarray}
where $P_{0,0}$, $P_{mix}$ and $P_{1,1}$ are the probability measures with probability densities $g_{0}$, $g_{1}$ and $g_{2}$, respectively. The proposed scheme converts the trinary hypothesis test into a composite binary hypothesis test that $H_{0,0}$ versus $\{H_{mix}, H_{1,1}\}$. Since $H_{mix}$ is closer to $H_{0,0}$, we use the worst case likelihood ratio $g_{1}/g_{0}$ in the test. The observer switches to observe the new sequences if the test result favors $H_{0,0}$. Otherwise, the observer enters the refinement stage. In the refinement stage, the observer examines only one sequence, and there are only two possible outcomes: $H_{0}$ and $H_{1}$. In the proposed low-complexity scheme, the observer adopts SPRT to examine $s_{\tau_{0}}^{1}$ sequentially. Then the observer picks $s_{\tau_{0}}^{1}$ if the test favors $H_{1}$ and picks $s_{\tau_{0}}^{2}$ if the test favors $H_{0}$.
\end{rmk}

Recall that $\tau_{0}$ is a renewal process and $\tau_{0} = \sum_{n=1}^{N} \eta_{m,n}$, where $N$ is the total number of repetitions. With a little abuse of notation, we still use $\chi_{0}$ to denote the event that the observer switches to observe new sequences. In this scenario, $\chi_{0}$ can be written as $\left\{\tilde{p}_{\eta_{m}} > p_{U}\right\}$. Hence, $N$ is geometrically distributed:
\begin{eqnarray}
P(N=n) = [1-P(\chi_{0})][P(\chi_{0})]^{n-1}, \no
\end{eqnarray}
and
\begin{eqnarray}
\mathbb{E}[N] = \frac{1}{1-P(\chi_{0})}.
\end{eqnarray}
Let $\alpha_{m} := 1-P_{0,0}(\chi_{0})$, $\beta_{m} := P_{mix}(\chi_{0})$ and $\gamma_{m} := P_{1,1}(\chi_{0})$, we have
\begin{eqnarray}
P(\chi_{0}) &=& (1-\pi_{0})^{2}P_{0,0}(\chi_{0}) + 2(1-\pi_{0})\pi_{0} P_{mix}(\chi_{0}) + \pi_{0}^{2} P_{1,1}(\chi_{0})\no\\
&=&  p_{0}^{1,1}(1-\alpha_{m}) + p_{0}^{mix}\beta_{m} + p_{0}^{1,1}\gamma_{m}.
\end{eqnarray}

\subsection{Analysis of the False Identification Probability} \label{sec:mixed_fip}
In this subsection, we discuss how to choose $\pi_{L}$, $\pi_{U}$, $A_{s}$ and $B_{s}$ so that the FIP constraint is satisfied. This will facilitate our further analysis of the average search delay in Section \ref{sec:mixed_asd}. In the refinement stage, we denote
\begin{eqnarray}
&&\alpha_{sprt} := P(s_{\tau_{0}}^{1} \text{ is claimed to be } f_{1} | s_{\tau_{0}}^{1} \text{ is generated by } f_{0}), \no\\
&&\gamma_{sprt} := P(s_{\tau_{0}}^{1} \text{ is claimed to be } f_{0} | s_{\tau_{0}}^{1} \text{ is generated by } f_{1}) \no
\end{eqnarray}
as Type I error and Type II error, respectively. According to the proposition of SPRT \emph{(\cite{Poor:Book:08}, Proposition 4.10)}, we have
\begin{eqnarray}
&&\alpha_{sprt} \leq B_{r}^{-1}(1- \gamma_{sprt}) < B_{r}^{-1}, \no\\
&&\gamma_{sprt} \leq A_{r}(1- \alpha_{sprt}) < A_{r}. \no
\end{eqnarray}
Hence, the probability that the observer infers the state of $s_{\tau_{0}}^{1}$ incorrectly is
\begin{eqnarray}
P(s_{\tau_{0}}^{1} \text{ is inferred incorrectly}) &=& q_{1,0}\gamma_{sprt} + \left(1-q_{1,0}\right)\alpha_{sprt} \no\\
&<& q_{1,0}A_{r} + \left(1-q_{1,0}\right)B_{r}^{-1} \no\\
&=& q_{1,0}\frac{1-q_{1,0}}{q_{1,0}}\frac{q_{L}}{1-q_{L}} + \left(1-q_{1,0}\right)\frac{q_{1,0}}{1-q_{1,0}}\frac{1-q_{U}}{q_{U}}\no\\
&=& \left(1-q_{1,0}\right)\frac{q_{L}}{1-q_{L}} + q_{1,0}\frac{1-q_{U}}{q_{U}}.
\end{eqnarray}
If we choose
$$q_{L}=\frac{\zeta/2}{1+\zeta/2} \text{ and }  q_{U}=\frac{1}{1+\zeta/2},$$
then we have
\begin{eqnarray}
P(s_{\tau_{0}}^{1} \text{ is inferred incorrectly}) < \left(1-q_{1,0}\right)\frac{\zeta}{2} + q_{1,0}\frac{\zeta}{2} = \frac{\zeta}{2}. \label{eq:err_refine}
\end{eqnarray}
Notice that under this selection of $q_{L}$ and $q_{U}$, $P(s_{\tau_{0}}^{1} \text{ is inferred incorrectly})$ is independent of $q_{1,0}$.

When the observer enters the refinement stage, there are three possible cases: 1) Both $s_{\tau_{0}}^{1}$ and $s_{\tau_{0}}^{2}$ are generated by $H_{1}$. In this case one can make any decision without causing an identification error; 2) Both $s_{\tau_{0}}^{1}$ and $s_{\tau_{0}}^{2}$ are generated by $H_{0}$. In this case an identification error will occur no matter what decision is made; 3) One of $s_{\tau_{0}}^{1}$ and $s_{\tau_{0}}^{2}$ is generated by $H_{0}$, while the other is by $H_{1}$. In this case the false identification probability is described in \eqref{eq:err_refine}. Now, we define three events:
\begin{eqnarray}
&&E_{1,1} = \{ \text{both } s_{\tau_{0}}^{1} \text{ and } s_{\tau_{0}}^{2} \text{ are generated from } f_{1} \}; \no\\
&&E_{mix} = \{ \text{one of }s_{\tau_{0}}^{1} \text{ and } s_{\tau_{0}}^{2} \text{ is generated from } f_{1}, \text{ the other is generated from } f_{0} \}; \no\\
&&E_{0,0} = \{ \text{both } s_{\tau_{0}}^{1} \text{ and } s_{\tau_{0}}^{2} \text{ are generated from } f_{0} \}. \no
\end{eqnarray}
Correspondingly, we have
\begin{eqnarray}
&&P(H^{\delta} = H_{0}|E_{1, 1}, N=1) = 0; \no\\
&&P(H^{\delta} = H_{0}|E_{mix}, N=1) < \zeta/2; \no\\
&&P(H^{\delta} = H_{0}|E_{0, 0}, N=1) = 1. \no
\end{eqnarray}
By Bayes' rule, it is easy to verify that
\begin{eqnarray}
P(E_{0, 0}|N=1) = \frac{p_{0}^{0,0} \alpha_{m}}{p_{0}^{0,0} \alpha_{m} + p_{mix}^{0,0}(1-\beta_{m}) + p_{0}^{1,1}(1-\gamma_{m})}. \no
\end{eqnarray}
Therefore, we have
\begin{eqnarray}
P(H^{\delta} = H_{0}) &=& \sum_{n=1}^{\infty} P(H^{\delta} = H_{0}|N=n)P(N=n) \no\\
&=& P(H^{\delta} = H_{0}|N=1)\sum_{n=1}^{\infty}P(N=n) \no\\
&=& P(H^{\delta} = H_{0}|N=1) \no \\
&=& P(H^{\delta} = H_{0}|E_{1, 1}, N=1)P(E_{1, 1}|N=1) \no\\
&+& P(H^{\delta} = H_{0}|E_{mix}, N=1)P(E_{mix}|N=1)\no\\
&+& P(H^{\delta} = H_{0}|E_{0, 0}, N=1)P(E_{0, 0}|N=1) \no\\
&<& \frac{\zeta}{2} + \frac{p_{0}^{0,0} \alpha_{m}}{p_{0}^{0,0} \alpha_{m} + p_{mix}^{0,0}(1-\beta_{m}) + p_{0}^{1,1}(1-\gamma_{m})}. \no
\end{eqnarray}
Hence, in order to satisfy the FIP constraint, we need to design threshold $A_{s}$, $B_{s}$ such that
\begin{eqnarray}
\frac{p_{0}^{0,0} \alpha_{m}}{p_{0}^{0,0} \alpha_{m} + p_{0}^{mix}(1-\beta_{m}) + p_{0}^{1,1}(1-\gamma_{m})} < \frac{\zeta}{2}. \label{eq:ineq1}
\end{eqnarray}
As $\pi_{0}\rightarrow 0$, \eqref{eq:ineq1} is equivalent to
\begin{eqnarray}
\alpha_{m} &<& \frac{p_{0}^{mix}(1-\beta_{m}) + p_{0}^{1,1}(1-\gamma_{m})}{p_{0}^{0,0}} \frac{\zeta/2}{1-\zeta/2} \no\\
&=& (1-\beta_{m})\frac{\pi_{0}}{1-\pi_{0}}\frac{\zeta}{1-\zeta/2} + (1-\gamma_{m})\frac{\pi_{0}^{2}}{(1-\pi_{0})^2}\frac{\zeta/2}{1-\zeta/2} \no \\
&=& (1-\beta_{m})\frac{\pi_{0}}{1-\pi_{0}}\frac{\zeta}{1-\zeta/2}(1+o(1)). \label{eq:double_alpha}
\end{eqnarray}
Using these results, we have the following conclusion regarding the selection of the thresholds.
\begin{cor}
\begin{eqnarray}
&& q_{L}=\frac{\zeta/2}{1+\zeta/2}, \quad q_{U}=\frac{1}{1+\zeta/2}, \no\\
&& A_{s} = 1, \quad B^{-1}_{s} = \frac{\pi_{0}}{1-\pi_{0}}\frac{\zeta}{1-\zeta/2}(1+o(1)) \label{eq:double_threshold}
\end{eqnarray}
is a set of thresholds that satisfy the FIP constraint.
\end{cor}
\begin{proof}
We only need to show that the selection of $A_{s}$ and $B_{s}$ can achieve \eqref{eq:double_alpha}. Since the CUSUM test is used in the scanning stage, the proof of this statement is similar to the proofs in Theorem \ref{thm:single_threshold} and Corollary \ref{cor:single_threshold}. Specifically, if $\alpha$, $\beta$, $\eta$ and $B$ used in Theorem \ref{thm:single_threshold} and Corollary \ref{cor:single_threshold} are replaced by $\alpha_{m}$, $\beta_{m}$, $\eta_{m}$ and $B_{s}$ respectively, and replace the random walk by $W_{k}^{(m)} = \sum_{i=1}^{k} g_{1}(Z_{i})/g_{0}(Z_{i})$, we can obtain
\begin{eqnarray}
\alpha_{m} < (1-\beta_{m}) B_{s}^{-1}. \no
\end{eqnarray}
As the result, if we choose $B_{s}$ as shown in the corollary, \eqref{eq:double_alpha} is satisfied. For simplicity, we set $A_{s}=1$.
\end{proof}

\subsection{Analysis of the Average Search Delay} \label{sec:mixed_asd}
As mentioned in Section \ref{subsec:single_performance}, we consider two cases, the FIE case and the RIE case, in the asymptotic analysis as $\pi_{0} \rightarrow 0$. Let $\rho_{m} = \alpha_{m}/\pi_{0}$. By \eqref{eq:double_alpha}, it is easy to see that $\rho_{m}$ is a constant in the FIE case and $\rho_{m} \rightarrow 0$ in the RIE case.

We first consider the delay caused in the refinement stage. For the FIE case, the thresholds $q_{L}$ and $q_{U}$ for SPRT are constants since $\zeta$ is a constant within $(0, 1)$. Therefore, in this case the expected delay is  finite, i.e., $\mathbb{E}_{q_{1,0}}[\tau_{1}] < \infty$. For the RIE case, we have the following lemma on the delay in the refinement stage:
\begin{lem} \label{lem:refine_delay}
If $\zeta \rightarrow 0$, then for any given $q_{1,0} \in (0, 1)$,
$$\mathbb{E}_{q_{1,0}}[\tau_{1}] = |\log \zeta| (1+o(1)).$$
\end{lem}
\begin{proof}
From \eqref{eq:double_threshold} we know that $q_{L} \rightarrow 0$ and $q_{U} \rightarrow 1$ as $\zeta \rightarrow 0$. Then the Type I error of SPRT in the refinement stage can be approximated by \emph{(\cite{Poor:Book:08}, Proposition 4.10)}:
\begin{eqnarray}
\alpha_{sprt} &\approx& \frac{1-A_{r}}{B_{r}-A_{r}} \rightarrow \frac{\zeta}{2}\frac{q_{1,0}}{1-q_{1,0}}, \no
\end{eqnarray}
and the Type II error can be approximate by
\begin{eqnarray}
\gamma_{sprt} \approx \frac{B_{r}-1}{B_{r}-A_{r}} A_{r} \rightarrow \frac{\zeta}{2}\frac{q_{1,0}}{1-q_{1,0}}. \no
\end{eqnarray}
Hence, the delay caused by SPRT is given as \emph{(\cite{Poor:Book:08},  Proposition 4.11)}:
\begin{eqnarray}
\mathbb{E}_{0}[\tau_{1}] \hspace{-3mm}&=&\hspace{-3mm} \frac{-1}{D(f_{0}||f_{1})}\left[ \alpha_{sprt}\log\frac{1-\gamma_{sprt}}{\alpha_{sprt}}+\left(1- \alpha_{sprt}\right)\log\frac{\gamma_{sprt}}{1-\alpha_{sprt}}\right] \no\\
\hspace{-3mm}&=&\hspace{-3mm} \frac{1-2\alpha_{sprt}}{D(f_{0}||f_{1})} \log \frac{1-\alpha_{sprt}}{\alpha_{sprt}} \no\\
\hspace{-3mm}&=&\hspace{-3mm} |\log \zeta/2|(1+o(1)), \no\\
\mathbb{E}_{1}[\tau_{1}] \hspace{-3mm}&=&\hspace{-3mm} \frac{1}{D(f_{1}||f_{0})}\left[ \left(1-\gamma_{sprt}\right)\log\frac{1-\gamma_{sprt}}{\alpha_{sprt}}+\gamma_{sprt}\log\frac{\gamma_{sprt}}{1-\alpha_{sprt}}\right] \no\\
\hspace{-3mm}&=&\hspace{-3mm} \frac{1-2\alpha_{sprt}}{D(f_{1}||f_{0})} \log \frac{1-\alpha_{sprt}}{\alpha_{sprt}} \no\\
\hspace{-3mm}&=&\hspace{-3mm} |\log \zeta/2|(1+o(1)). \no
\end{eqnarray}
Since $\mathbb{E}_{q_{1,0}}[\tau_{1}]$ is a linear combination of $\mathbb{E}_{0}[\tau_{1}]$ and $\mathbb{E}_{1}[\tau_{1}]$, we have
$$ \min\{ \mathbb{E}_{0}[\tau_{1}], \mathbb{E}_{1}[\tau_{1}] \} \leq \mathbb{E}_{q_{1,0}}[\tau_{1}] \leq \max\{ \mathbb{E}_{0}[\tau_{1}], \mathbb{E}_{1}[\tau_{1}] \}.$$
Therefore
$$ \mathbb{E}_{q_{1,0}}[\tau_{1}] = |\log \zeta| (1+o(1)).$$
\end{proof}

\begin{rmk}
As we can see from the above lemma, if $\zeta \rightarrow 0$, the asymptotic delay in the refinement stage is determined by $\zeta$, regardless of the value of $q_{1,0}$. Hence, for the implementation purpose, we can simply set $q_{1,0} = 1/2$. On the other hand, if $\zeta$ is a constant within $(0, 1)$,  the delay in the refinement stage is a finite value for any $q_{1,0} \in (0, 1)$. Hence, we can also set $q_{1,0} = 1/2$. As we will show later, compared with the delay in the scanning stage, the delay in the refinement stage is negligible. 
\end{rmk}

In the following, we study the delay incurred in the scanning stage. Recall that $\tau_{0} = \sum_{n=1}^{N} \eta_{m,n}$, by applying the Wald's identity, we obtain
\begin{eqnarray}
\mathbb{E}[\tau_{0}] &=& \mathbb{E}[N] \mathbb{E}[\eta_{m}] = \frac{\mathbb{E}[\eta_{m}]}{1-P(\chi_{0})} \no\\
&=&\frac{p_{0}^{0,0}\mathbb{E}_{0,0}[\eta_{m}]+p_{0}^{mix}\mathbb{E}_{mix}[\eta_{m}]+ p_{0}^{1,1}\mathbb{E}_{1,1}[\eta_{m}]}{p_{0}^{0,0}\alpha_{m}+p_{0}^{mix}(1-\beta_{m})+p_{0}^{1,1}(1-\gamma_{m})}. \no
\end{eqnarray}

As $\pi_{0} \rightarrow 0$, it is easy to verify that
\begin{eqnarray}
\mathbb{E}[\tau_{0}] &\rightarrow& \frac{1-\pi_{0}}{\rho_{m}(1-\pi_{0}) + 2(1-\beta_{m})}\frac{1}{\pi_{0}}\mathbb{E}_{0,0}[\eta_{m}] \no\\
&+& \frac{2(1-\pi_{0})}{(1-\pi_{0}^2)\rho_{m}+2(1-\pi_{0})(1-\beta_{m})}\mathbb{E}_{mix}[\eta_{m}] + \pi_{0}\mathbb{E}_{1,1}[\eta_{m}].\label{eq:scan_delay}
\end{eqnarray}

\begin{lem} \label{lem:scan_delay}
If $0< D(g_{1}||g_{0})<\infty$ and $0 < \mathbb{E}_{1,1}[\log(g_{1}/g_{0})]<\infty$, then as $\pi_{0} \rightarrow 0$, the scanning stage delay for the FIE case is given as
\begin{eqnarray}
\mathbb{E}[\tau_{0}] = \frac{1-\pi_{0}}{\rho_{m}(1-\pi_{0})+2(1-\beta_{m})}\frac{1}{\pi_{0}}\mathbb{E}_{0,0}[\eta_{m}](1+o(1)). \label{eq:regular_delay}
\end{eqnarray}
In addition, if $\pi_{0}|\log \zeta| \rightarrow 0$, then the scanning stage delay for the RIE case is given as
\begin{eqnarray}
\mathbb{E}[\tau_{0}] = \frac{1-\pi_{0}}{(1-\beta_{m})}\frac{1}{2\pi_{0}}\mathbb{E}_{0,0}[\eta_{m}](1+o(1)). \label{eq:rare_delay}
\end{eqnarray}
\end{lem}
\begin{proof}
By the argument in Proposition \ref{lem:delaybound}, we immediately have
\begin{eqnarray}
&& 0 < \mathbb{E}_{0,0}[\eta_{m}] < \infty, \no\\
&& 0 < \mathbb{E}_{mix}[\eta_{m}] \leq (1-\beta_{m})\frac{|\log B_{s}|}{D(g_{1}||g_{0})}(1+o(1)). \no
\end{eqnarray}
In the following, we only need to study $\mathbb{E}_{1,1}[\eta_{m}]$. With a little abuse of notation, we denote
\begin{eqnarray}
W_{k}^{(m)} := \sum_{i=1}^{k} \log \frac{g_{1}(Z_{i})}{g_{0}(Z_{i})}.  \no
\end{eqnarray}
Since $A_{s}=1$ and $B_{s}^{-1}=\frac{\pi_{0}}{1-\pi_{0}}\frac{\zeta}{1-\zeta/2}(1+o(1))$, we have
\begin{eqnarray}
\mathbb{E}_{1,1}\left[W_{k}^{(m)}\right] &=& \mathbb{E}_{1,1}\left[W_{k}^{(m)}\Big|W_{k}^{(m)} \leq 0\right] P_{1,1}\left(W_{k}^{(m)} \leq 0 \right) \no\\
&+& \mathbb{E}_{1,1}\left[W_{k}^{(m)}\Big|W_{k}^{(m)} \geq \log B_{s}\right] P_{1,1}\left(W_{k}^{(m)} \geq \log B_{s} \right) \no\\
&<& \mathbb{E}_{1,1}\left[W_{k}^{(m)}\Big|W_{k}^{(m)} \geq \log B_{s} \right] P_{1,1}\left(W_{k}^{(m)} \geq \log B_{s} \right) \no\\
&=& (1-\gamma_{m})|\log B_{s}|(1+o(1)). \no
\end{eqnarray}
At the same time, by Wald's identity, we have
\begin{eqnarray}
\mathbb{E}_{1,1}\left[W_{k}^{(m)}\right] = \mathbb{E}_{1,1}[\eta_{m}]\mathbb{E}_{1,1}\left[\log \frac{g_{1}(Z_{1})}{g_{0}(Z_{1})} \right]. \no
\end{eqnarray}
Then we have
\begin{eqnarray}
\mathbb{E}_{1,1}[\eta_{m}] &=& \frac{\mathbb{E}_{1,1}\left[W_{k}^{(m)}\right]}{\mathbb{E}_{1,1}\left[\log \left(g_{1}(Z_{1})/g_{0}(Z_{1})\right)\right]} \no\\
&<& (1-\gamma_{m})\frac{|\log B_{s}|}{\mathbb{E}_{1,1}\left[\log \left(g_{1}(Z_{1})/g_{0}(Z_{1})\right)\right]}(1+o(1)). \no
\end{eqnarray}
Notice that $\pi_{0}|\log B_{s}| \rightarrow 0$ for both the FIE and RIE cases. Hence, the last term in \eqref{eq:scan_delay} goes to zero. Moreover, since the second term in \eqref{eq:scan_delay} is on the order of $O(|\log B_{s}|)$, then the first term in \eqref{eq:scan_delay}, which is on the order of $O(\pi_{0}^{-1})$, dominates the delay in the scanning stage. Therefore, we have
\begin{eqnarray}
\mathbb{E}[\tau_{0}] = \frac{1-\pi_{0}}{\rho_{m}(1-\pi_{0})+2(1-\beta_{m})}\frac{1}{\pi_{0}}\mathbb{E}_{0,0}[\eta_{m}](1+o(1)). \no
\end{eqnarray}
Then \eqref{eq:regular_delay} follows. \eqref{eq:rare_delay} follows from the fact that $\rho_{m} \rightarrow 0$ in the RIE case.
\end{proof}

\begin{thm}
With the assumptions in Lemma \ref{lem:scan_delay}, then as $\pi_{0} \rightarrow 0$, we have
\begin{eqnarray}
ASD_{m}^{*} &=& \mathbb{E}[\tau_{0}^{*} + \tau_{1}^{*}] \no\\
&\leq& \frac{1-\pi_{0}}{\rho_{m}(1-\pi_{0})+2(1-\beta_{m})}\frac{1}{\pi_{0}}\mathbb{E}_{0,0}[\eta_{m}](1+o(1)). \no
\end{eqnarray}
\end{thm}
\begin{proof}
This theorem follows from the fact that the scanning stage delay is on the order of $O(\pi_{0}^{-1})$, while the refinement stage delay is either a finite number (under the FIE case) or on the order of $O(|\log \zeta|)$ (under the RIE case). Since $\pi_{0}|\log \zeta| \rightarrow 0$, ASD is dominated by the detection delay in the scanning stage. Since the proposed low complexity search strategy is not optimal, the inequality holds.
\end{proof}

\begin{rmk}
The conclusion that the scanning stage delay dominates the refinement stage delay as $\pi_{0} \rightarrow 0$ is a very important insight we gain from the two-observation mixed search strategy. Intuitively, to maintain a low identification error, the observer needs to make sure that at least one of the two candidate sequence selected in the scanning stage is generated from $f_{1}$ since the observer cannot come back to the scanning stage again after it enters the refinement stage. In other words, the mistake made in the scanning stage cannot be compensated in the refinement stage. Hence, the observer spends a long time, which turns out to be in proportion to $\pi_{0}^{-1}$, in the scanning stage to make an accurate decision.
\end{rmk}

\subsection{Mixed Observation Strategy vs. Single Observation Strategy} \label{sec:mixed_gain}
In this subsection, we compare the performances of the mixed observation search strategy and the single observation search strategy as $\pi_{0} \rightarrow 0$. From results discussed above, we have the following proposition:

\begin{prop} \label{prop:gain}
\begin{eqnarray}
\frac{ASD_{m}^{*}}{ASD^{*}} &<& \frac{\rho(1-\pi_{0})+(1-\beta)}{\rho_{m}(1-\pi_{0})+2(1-\beta_{m})}\frac{\mathbb{E}_{0,0}[\eta_{m}]}{\mathbb{E}_{0}[\eta]} \no\\
&\rightarrow& \frac{1-\beta}{2(1-\beta_{m})}\frac{\mathbb{E}_{0,0}[\eta_{m}]}{\mathbb{E}_{0}[\eta]} \text{ as } \rho\rightarrow 0, \rho_{m}\rightarrow 0.
\end{eqnarray}
\end{prop}

For further analysis, one needs to characterize $\beta$, $\beta_{m}$, $\mathbb{E}_{0}[\eta]$ and $\mathbb{E}_{0,0}[\eta_{m}]$. These four quantities depend on the undershoot of the corresponding random walks crossing the lower bound. As the lower bound goes to zero, the asymptotic results of the undershoot have been discussed in \cite{Siegmund:ADAP:79}, \cite{Siegmund:Book:85}, \cite{Lotov:AnP:71}, \cite{Chang:AnP:97}, etc. However, these results cannot be used in our analysis since the lower bound $A_{s}$ is equal to $1$ in our case. In general, one needs to use numerical methods to estimate these four quantities. However, we can show that the delay ratio approaches $1/2$, i.e., the mixed search strategy reduces the search delay to the half of $ASD^{*}$, under some special cases.


\begin{exmpl} \label{eg:1}
Consider the quickest search problem for
$$ H_{0}: Y_{k}^{i} \sim \mathcal{N}(0, \sigma^2) \text{  vs. } H_{1}: Y_{k}^{i} \sim \mathcal{N}(\mu, \sigma^2), $$
where $\sigma$ is a finite constant. Then
$$\frac{ASD_{m}^{*}}{ASD^{*}} \rightarrow \frac{1}{2}$$
as $\pi_{0} \rightarrow 0$, $\mu \rightarrow \infty$.
\end{exmpl}
\begin{proof}
For the single observation search strategy, we assume that the observer observes sequence $s_{t}$ at time slot $i=1, 2, \ldots \eta$. Let $Y_{i}$ be the $i^{th}$ observation drawn from $s_{t}$, and let $W_{k}=\sum_{i=1}^{k} \log f_{1}(Y_{i})/f_{0}(Y_{i})$ be the random walk. Define $\bar{Y}_{k} = \frac{1}{k} \sum_{i=1}^{k} Y_{i}$. It is easy to verify that the event $\chi_{0}=\{W_{\eta} < 0\}$ can be equivalently written as $\chi_{0}=\{\bar{Y}_{\eta} < \mu/2\}$.

The distribution of $\bar{Y}_{k}$ is given as
$$f(\bar{y}_{k}|H_{0}) = \mathcal{N}(0, \sigma^2/k), \quad f(\bar{y}_{k}|H_{1}) = \mathcal{N}(\mu, \sigma^2/k). $$
As $\mu \rightarrow \infty$, one can verify that
\begin{eqnarray}
&& \beta = P_{1}(\chi_{0}) \rightarrow 0,  \no\\
&& \mathbb{E}_{0}[\eta] = \sum_{k=1}^{\infty} k P_{0}(\eta=k) \rightarrow 1, \no
\end{eqnarray}

For the mixed observation search strategy, the proposed low complexity algorithm conduct following test in the scanning stage,
$$ H_{0,0}: Z_{k} \sim \mathcal{N}(0, 2\sigma^2) \text{  vs. } H_{mix}: Z_{k} \sim \mathcal{N}(\mu, 2\sigma^2). $$
This test is the same as the above one, hence we can obtain $\beta_{m} \rightarrow 0$ and  $\mathbb{E}_{0,0}[\eta_{m}] \rightarrow 1$ as $\mu \rightarrow \infty$.

As $\pi_{0} \rightarrow 0$, we can use Proposition \ref{prop:gain}, and the delay ratio goes to $1/2$.
\end{proof}
\begin{rmk}
For this particular binary test, one can write the expressions of $\beta$, $\beta_{m}$, $\mathbb{E}_{0}[\eta]$ and $\mathbb{E}_{0,0}[\eta_{m}]$ explicitly using the Gaussian distribution. However, the expressions are long and complex. In Section \ref{sec:simulation}, we conduct numerical simulations for the Gaussian case with different values of mean and variance, the results show that the delay ratio is close to $1/2$ even for small $\mu$'s.
\end{rmk}

\begin{rmk}
Following the discussion in Example \ref{eg:1}, one can show that the delay ratio also approaches to $1/2$ for some other \emph{infinitely divisible} distributions. For example, the quickest search problem for gamma distribution
$$H_{0}: Y_{k}^{i} \sim \Gamma(\kappa_{0}, \theta) \text{ vs. } H_{1}: Y_{k}^{i} \sim \Gamma(\kappa_{1}, \theta)$$
when $|\kappa_{1} - \kappa_{0}| \rightarrow \infty$, and the quickest search problem for Poisson distribution
$$H_{0}: Y_{k}^{i} \sim Poisson(\lambda_{0}) \text{ vs. } H_{1}: Y_{k}^{i} \sim Possion(\lambda_{1})$$
when $|\lambda_{1} - \lambda_{0}| \rightarrow \infty$ will also achieve the delay ratio $1/2$. The infinitely divisible distributions is relatively easy to analyze since the sum of i.i.d. random variables has the same type of distribution as each individual random variable only with different parameters.
\end{rmk}

\section{Extension: $2$-based Multi-stage Search Strategy} \label{sec:extension}
In this section, we extend the mixed search strategy to multiple sequences. The proposed new search strategy is named as $2$-based multi-stage search strategy.

This search strategy consists of one scanning stage and $K$ refinement stages. Specifically, in the scanning stage, the observer observes the sum of samples from $2^{K}$ sequences. After taking each observation, the observer has to make one of the following three decisions: 1) to take another observation from the same group of sequences; 2) to switch to another group of $2^{K}$ sequences; or 3) to stop scanning and enter the first refinement stage. Hence, there are $2^{K}$ candidate sequences for the first refinement stage.

Each refinement stage selects half of the candidate sequences for the next refinement stage. Specifically, the $i^{th}$ refinement stage, $i=1, 2, \ldots, K$, has $2^{K-i+1}$ candidate sequences, and the observer divides them equally into two groups. Then the observer observes the sum of samples from the sequences in the first group. After taking each observation, the observer has to make one of the following two decisions: 1) to take another observation from the first group; or 2) to stop the $i^{th}$ refinement stage and select one of the two groups for the next refinement stage. Hence, there are $2^{K-i}$ candidate sequences left after the $i^{th}$ refinement stage. When $i=K$, i.e., after the last refinement stage, there is only one sequence left, which will be claimed to be generated from $f_{1}$.

We use $\tau_{0}$ to denote the stopping time for the scanning stage, $\tau_{i}$ to denote the stopping time for the $i^{th}$ refinement stage. We still use $\phiv$ to denote the sequence of switching rules used in the scanning stage. Furthermore, we use $\delta^{(i)}$ to denote the terminal decision rule used in the $i^{th}$ refinement stage. In particular, $\{\delta^{(i)}=1\}$ implies that the observer selects the sequences in the first group at the end of the $i^{th}$ refinement stage, and $\{\delta^{(i)}=2\}$ implies that the sequences in the second group is selected.


Our goal is to characterize the following optimization problem
\begin{eqnarray}
w_{0}:= \inf_{\tau_{0}, \phiv, \tau_{1}, \delta^{(1)}, \ldots, \tau_{K}, \delta^{(K)}} c\mathbb{E}\left[ \sum_{i=0}^{K} \tau_{i} \right]+P\left( H^{\delta^{(K)}} = H_{0} \right).
\end{eqnarray}

Similar to Section \ref{sec:mixed}, we convert this multiple stopping time problem into $K+1$ concatenated single stopping time problems. Let $Z_{k}$ be the $k^{th}$ observation taken in the scanning stage, $X_{j}^{(i)}$ be the $j^{th}$ observation taken in the $i^{th}$ refinement stage. Moreover, let $\mathcal{F}_{k} = \sigma\{Z_{1}, \ldots, Z_{k}\}$ and $\mathcal{G}_{j}^{(i)} = \sigma\{ Z_{1}, \ldots, Z_{\tau_{0}}, X_{1}^{(1)}, \ldots, X_{\tau_{1}}^{(1)}, \ldots, X_{1}^{(i)}, \ldots, X_{j}^{(i)}\}$, for $i=1, \ldots, K$. Corresponding to Lemma \ref{lem:multi_stopping}, we have the following lemma for this multi-stage search strategy.
\begin{lem}
For $i=2,\ldots, K-1$, let
\begin{eqnarray}
v^{(K)}\left(\mathcal{G}^{(K-1)}_{\tau_{K-1}} \right) &=& v^{(K)}(\tau_{0}, \phiv, \tau_{1}, \delta^{(1)}, \ldots, \tau_{K-1}, \delta^{(K-1)})\no\\
&:=& \inf_{\tau_{K}, \delta^{(K)}} \mathbb{E}\left[ c\tau_{K} + \mathbf{1}_{\{ H^{\delta^{(K)}} = H_{0} \}} \Big| \mathcal{G}^{(K-1)}_{\tau_{K-1}} \right],  \no\\
v^{(i)}\left(\mathcal{G}^{(i-1)}_{\tau_{i-1}} \right) &=& v^{(i)}(\tau_{0}, \phiv, \tau_{1}, \delta^{(1)}, \ldots, \tau_{i-1}, \delta^{(i-1)}) \no\\
&:=& \inf_{\tau_{i}, \delta^{(i)}} \mathbb{E}\left[ c\tau_{i}+w^{(i+1)}(\tau_{0}, \phiv, \tau_{1}, \delta^{(1)}, \ldots, \tau_{i}, \delta^{(i)}) \Big| \mathcal{G}^{(i-1)}_{\tau_{i-1}} \right], \no\\
v^{(1)}(\mathcal{F}_{\tau_{1}}) &=& v^{(1)}(\tau_{0}, \phiv) \no\\
&:=& \inf_{\tau_{1}, \delta^{(1)}} \mathbb{E}\left[ c\tau_{1}+w^{(2)}(\tau_{0}, \phiv, \tau_{1}, \delta^{(1)}) \Big| \mathcal{F}_{\tau_{0}} \right], \no \\
u_{0} &:=& \inf_{\tau_{0}, \phiv} \mathbb{E}\left[ c\tau_{0}+w^{(1)}(\tau_{0}, \phiv) \right], \no
\end{eqnarray}
then
$$u_{0}=w_{0}.$$
\end{lem}
\begin{proof}
This proof follows the similar steps in the proof for Lemma \ref{lem:multi_stopping}. Hence, we omit the details for brevity.
\end{proof}

In definitions above, $w^{(i)}$ can be viewed as the cost function associated with the $i^{th}$ refinement stage, and $u_{0}$ is the cost function associated with the scanning stage. This set of concatenated single stopping time problems can be solved using the backward induction. Specifically, $w^{(i)}(\mathcal{G}^{(i-1)}_{\tau_{i-1}} )$ can be written as
\begin{eqnarray}
w^{(i)}\left(\mathcal{G}^{(i-1)}_{j} \right) &=& \min\left\{ w^{(i+1)}\left( \mathcal{G}_{j}^{(i-1)} \right), c+\mathbb{E}\left[w^{(i)}\left(\mathcal{G}^{(i-1)}_{j+1}\right)|\mathcal{G}^{(i-1)}_{j} \right] \right\} \no\\
&=& \min\left\{ w^{(i+1)}\left( \mathcal{G}_{j}^{(i-1)} \right), c+ \min\left\{\Delta_{1}^{(i)}, \Delta_{2}^{(i)} \right\}\right\}, \no
\end{eqnarray}
where
\begin{eqnarray}
\Delta_{1}^{(i)} &=& \mathbb{E}\left[w^{(i)}\left(\mathcal{G}^{(i-1)}_{j+1}\right)|\mathcal{G}^{(i-1)}_{j}, \delta^{(i)}=1 \right],\no\\
\Delta_{2}^{(i)} &=& \mathbb{E}\left[w^{(i)}\left(\mathcal{G}^{(i-1)}_{j+1}\right)|\mathcal{G}^{(i-1)}_{j}, \delta^{(i)}=2 \right].\no
\end{eqnarray}
Hence, the optimal solution for the $i^{th}$ refinement stage is given as
\begin{eqnarray}
\tau_{i} &=& \inf\left\{ j \geq 0:  w^{(i+1)}\left( \mathcal{G}_{j}^{(i-1)} \right) \leq w^{(i)}\left(\mathcal{G}^{(i-1)}_{j} \right) \right\}, \no\\
\delta^{(i)} &=& \left\{
                              \begin{array}{cc}
                                 1 & \text{ if } \Delta_{1}^{(i)} \leq \Delta_{2}^{(i)}\\
                                 2 & \text{ otherwise } \\
                               \end{array}
\right..
\end{eqnarray}
Similarly, for the scanning stage, we have $u_{0} = U(\mathcal{F}_{0})$, in which $U(\cdot)$ satisfies the following recursion
\begin{eqnarray}
U(\mathcal{F}_{k}) &=& \min\{w^{(1)}(\mathcal{F}_{k}), c+\mathbb{E}[U(\mathcal{F}_{k+1})|\mathcal{F}_{k}]\} \no\\
&=& \min\{w^{(1)}(\mathcal{F}_{k}), c+\min\left\{ \Phi_{c}(\mathcal{F}_{k}), \Phi_{s}(\mathcal{F}_{k})\right\} \},\no
\end{eqnarray}
where
\begin{eqnarray}
\Phi_{c}(\mathcal{F}_{k}) &=& \mathbb{E}[U(\mathcal{F}_{k+1})|\mathcal{F}_{k}, \phi_{k}=0], \no\\
\Phi_{s}(\mathcal{F}_{k}) &=& \mathbb{E}[U(\mathcal{F}_{k+1})|\mathcal{F}_{k}, \phi_{k}=1], \no
\end{eqnarray}
and the optimal solution for the scanning stage is given as
\begin{eqnarray}
\tau_{1} &=& \inf\{k \geq 0:  w^{(1)}(\mathcal{F}_{k}) \leq U(\mathcal{F}_{k}) \}, \no\\
\phi_{k} &=& \left\{
                              \begin{array}{cc}
                                 0 & \text{ if } \Phi_{c}(\mathcal{F}_{k}) \leq \Phi_{s}(\mathcal{F}_{k})\\
                                 1 & \text{ otherwise } \\
                               \end{array}
\right..
\end{eqnarray}

\begin{rmk}
To propose an efficient low complexity algorithm is a challenging task for the multi-stage search strategy. Theoretically, one can follow the similar steps in Section \ref{sec:mixed_asym} to design a similar low complexity algorithm. For example, one can use the worst case sequential likelihood ratio $\prod_{i=1}^{k} g_{1}(Z_{i})/g_{0}(Z_{i})$ in the scanning stage, where $g_{0}$ is the pdf of the observation when all $2^{K}$ observing sequences are generated from $f_{0}$, and $g_{1}$ is the pdf when only one of the $2^{K}$ observing sequences is generated from $f_{1}$. However, this low complexity algorithm is not very efficient when $K$ is large because the procedure of mixing $2^{K}$ sequences leads to a small difference between $g_{0}$ and $g_{1}$ even when $f_{0}$ and $f_{1}$ have reasonable KL distance.
\end{rmk}

\section{Simulation} \label{sec:simulation}
In this section, we give some numerical examples to illustrate the analytical results of this paper. In the first two simulations, we illustrate the cost function of the two-observation mixed search strategy analyzed in Section \ref{sec:mixed}. In these numerical examples, we assume $f_{0} \sim \mathcal{N}(0, \sigma_{0}^2)$ and $f_{1} \sim \mathcal{N}(0, \sigma_{1}^2)$.

In the first simulation, we illustrate the cost function of the refinement stage $v\left(p^{1,1}, p^{mix}\right)$. In this simulation, we choose $\pi = 0.05$, $c=0.01$, $\sigma_{0}^2 = 1$ and $\sigma_{1}^2 = 3$. The simulation result is shown in Figure \ref{fig:cost_refine}. This simulation confirms our analysis that $v\left(p^{1,1}, p^{mix}\right)$ is a concave function within $[0, 1]$ over $\mathcal{P}$. We also notice that $v(1,0) = 0$ and $v(0, 0)=1$. This is reasonable since if the observer knows that both $s_{\tau_{0}}^{1}$ and $s_{\tau_{0}}^{2}$ are generated by $f_{1}$, which corresponds to $p^{1,1}_{\tau_{0}} = 1$ and $p^{mix}_{\tau_{0}} = 0$, the observer can make a decision on either of sequences without taking any further observation and making any error, hence the cost on the refinement stage is $0$. Similarly, if the observer knows that neither $s_{\tau_{0}}^{1}$ nor $s_{\tau_{0}}^{2}$ is generated by $f_{1}$, that is $p^{1,1}_{\tau_{0}} = 0$ and $p^{mix}_{\tau_{0}} = 0$, no matter what decision is made, the cost of error would be $1$. We also notice that in the area close to $(p^{1,1}_{\tau_{0}}, p^{mix}_{\tau_{0}}) = (0, 1)$, which indicates the observer is quite sure that one of sequences is generated by $f_{1}$, the cost is small. This is because the observer can significantly reduce the cost of the decision error by taking a few observations.

\begin{figure}[thb]
\centering
\includegraphics[width=0.45 \textwidth]{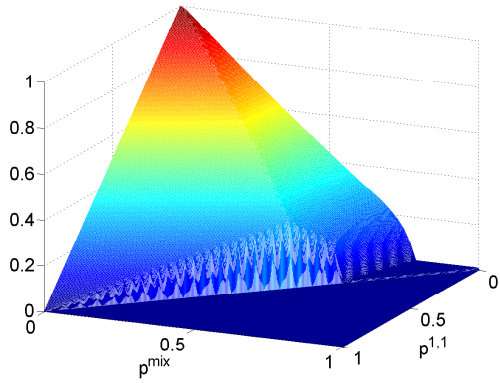}
\caption{An illustration of $v(p^{1,1}, p^{mix})$}
\label{fig:cost_refine}
\end{figure}

In the second simulation, we illustrate the overall cost function $U\left(p^{1,1}, p^{mix}\right)$ using the same simulation parameters. The simulation result is shown in Figure \ref{fig:cost_total}. This simulation confirms that $U\left(p^{1,1}, p^{mix}\right)$ is also a concave function over $\mathcal{P}$. Moveover, this function is flat at the top since it is upper bounded by a constant $c+\Phi_{s}$. This flat area corresponds to $R_{\phi}$, hence if $\left(p^{1,1}_{k}, p^{mix}_{k}\right)$ enters this region, the observer would switch to new sequences at time slot $k$. Similarly, the cost function is also upper bounded by $v\left(p^{1,1}, p^{mix}\right)$, which is shown in Figure \ref{fig:cost_refine}. On the region, $R_{\tau}$, where these two surfaces overlap each other, the observer would stop the scanning stage and enter the refinement stage. The locations of $R_{\phi}$ and $R_{\tau}$ are illustrated in Figure \ref{fig:regions}. In this figure, the left-lower half below the black solid line is the domain $\mathcal{P}$. The region circled by the red dash line is the sequence switching region $R_{\phi}$, and the region circled by the blue dot-dash line is the scanning stop region $R_{\tau}$. In this example, $R_{\tau}$ has two separate regions located around $(0, 1)$ and $(1,0)$ respectively, which means that the observer will enter the refinement stage as soon as it has enough confidence on that at least one of the observed sequences is generated by $f_{1}$. $R_{\tau}$ and $R_{\phi}$ can be computed offline.

\begin{figure}[thb]
\centering
\includegraphics[width=0.45 \textwidth]{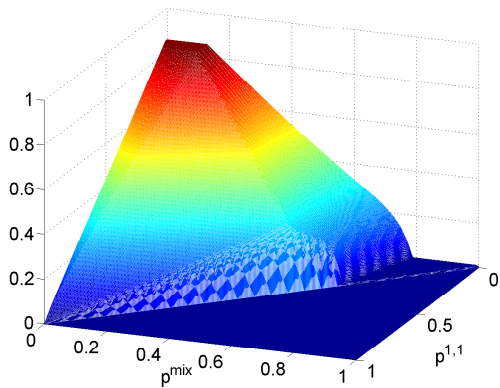}
\caption{An illustration of $V_{s}(p^{1,1}, p^{mix})$}
\label{fig:cost_total}
\end{figure}

\begin{figure}[thb]
\centering
\includegraphics[width=0.45 \textwidth]{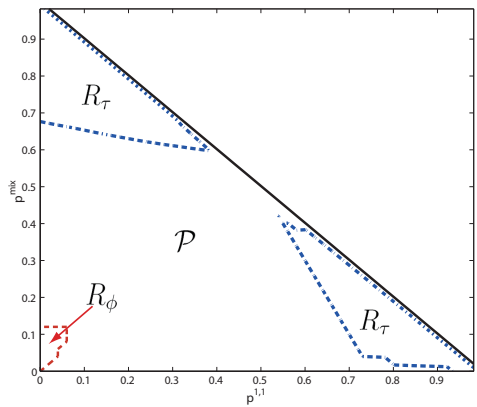}
\caption{The optimal stopping and switching regions}
\label{fig:regions}
\end{figure}

In the third and fourth simulation, we compare the performance of the optimal single observation strategy of \cite{Lai:TIT:11}, the optimal mixed observation strategy proposed in Section \ref{sec:mixed_opt} and the low complexity mixed observation strategy proposed in Section \ref{sec:mixed_asym}. In this set of simulations, we assume that $f_{0}$ is $\mathcal{N}(0, \sigma^{2})$ and $f_{1}$ is $\mathcal{N}(0, P+\sigma^{2})$. We further define the signal to noise ratio as $SNR = 10 \log (P / \sigma^{2})$.

The following simulation shows the performance of each search strategy under different SNR. In particular, we consider ASD of each search strategy by controlling FIP to be around $0.1$. The simulation results for $\pi_{0}=0.05$ are listed in Table \ref{tbl:ASDvsSNR} and are illustrated in Figure \ref{fig:asd_comparision}.  In Figure \ref{fig:asd_comparision}, the blue solid line, the red dot-dash line and the black dash line are ASDs of the optimal mixed observation search strategy, the low complexity mixed observation search strategy and the single observation search strategy, respectively. As we can see, both the optimal and the low complexity mixed observation search strategy outperform the single observation search strategy under this simulation setting. Moreover, the ASD of the low complexity algorithm approaches to that of the optimal mixed observation strategy as SNR increases. The corresponding delay ratio is shown in Figure \ref{fig:asd_reduction}, in which the blue solid line and the red dot-dash line are the delay ratios of the optimal and the low complexity mixed observation search strategy with respect to the single observation search strategy, respectively. Both of these two delay ratios decrease as SNR increases. The optimal mixing search strategy could save about $40\%$ search time. In the low SNR, the delay ratio of the low complexity algorithm is around $0.8$, which is not so significant. This is due to the fact that a small distance between $f_{1}$ and $f_{0}$ leads to a even smaller distance between $g_{1}$ and $g_{0}$, hence the detection in the scanning stage becomes challenging and the delay ratio is relatively large.

\begin{table}
\centering
\caption{The average search delay under different search strategies when $\pi_{0}=0.05$, $FIP \approx 0.1 $} \label{tbl:ASDvsSNR}
\begin{tabular}{|c|c|c|c|}
\hline
& \multicolumn{3}{|c|}{Average search delay } \\ \cline{2-4}
SNR & proposed optimal mixing search & proposed low complexity mixing search & single observation search \cite{Lai:TIT:11} \\ \hline
4 dB & 30.623 & 43.098 & 50.849 \\ \hline
6 dB & 26.545 & 31.055 & 40.353 \\ \hline
8 dB & 22.967 & 24.883 & 36.291\\ \hline
10 dB & 20.853 & 22.088 & 31.182 \\ \hline
12 dB & 17.751 & 18.503 & 28.798 \\ \hline
14 dB & 15.952 & 16.125 & 27.297 \\ \hline
16 dB & 14.766 & 15.204 &  25.045 \\ \hline
\end{tabular}
\end{table}

\begin{figure}[thb]
\centering
\includegraphics[width=0.55 \textwidth]{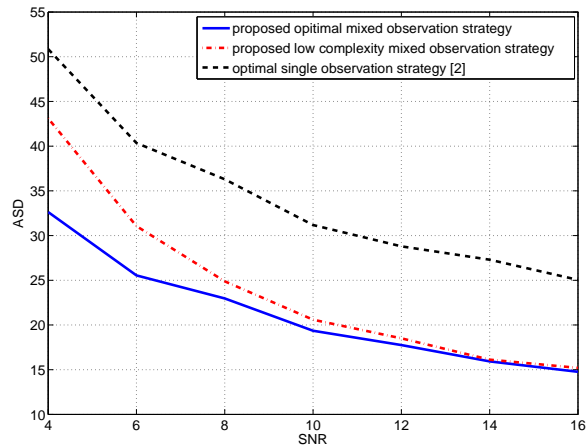}
\caption{ASD vs. SNR when $\pi_{0} = 0.05$, $FIP \approx 0.1$}
\label{fig:asd_comparision}
\end{figure}

\begin{figure}[thb]
\centering
\includegraphics[width=0.55 \textwidth]{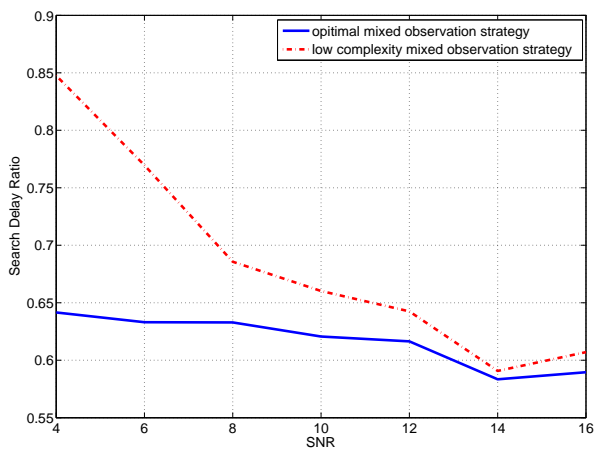}
\caption{The ASD reduction rate under different SNR when $\pi_{0} = 0.05$, $FIP \approx 0.1$}
\label{fig:asd_reduction}
\end{figure}

The fourth simulation illustrates the performance of each search strategy under different prior probability $\pi_{0}$. The simulation result for SNR$=12$dB is listed in table \ref{tbl:ASDvsPI}.  When $\pi_{0}$ is large, the mixed observation strategy does not have any advantages. Actually, in this case the performance of the single observation search strategy is slightly better than that of the mixed observation strategy. When the probability of the observing sequence generated by $f_{1}$ is large, the observer does not need to switch the observing sequence frequently; hence the observer could observe the sequence one by one. In the case, a second stage strategy is redundant. However, when $\pi_{0}$ is small, i.e., the majority of the sequences are generated by $f_{0}$, the mixed search strategy has advantage since it skips through the sequence generated by $f_{0}$ more efficiently.
\begin{table}
\centering
\caption{The average search delay under different prior probabilities} \label{tbl:ASDvsPI}
\begin{tabular}{|c|c|c|c|}
\hline
prior probability & \multicolumn{3}{|c|}{Average search delay } \\ \cline{2-4}
$\pi_{0}$ & optimal strategy & low complexity strategy & single observation \\ \hline
0.5 & 3.783 & 4.351 & 3.022 \\ \hline
0.3 & 5.662 & 5.945 & 5.012 \\ \hline
0.2 & 7.580 & 7.846 & 7.434 \\ \hline
0.1 & 10.923 & 12.293 & 15.886 \\ \hline
0.05 & 19.675 & 22.481 & 30.497 \\ \hline
0.01 & 90.131 & 96.867 & 149.991 \\ \hline
\end{tabular}
\end{table}

In the last two simulations, we illustrate the analytical results obtained in Section \ref{sec:mixed_asym}. In the following simulation, we first consider the delay ratio when both $f_{0}$ and $f_{1}$ are Gaussian distributions. In particular, we set $f_{0}$ to be $\mathcal{N}(0, 1)$ and $f_{1}$ to be $\mathcal{N}(\mu, \sigma^{2})$. Hence, every pair of $(\mu, \sigma)$ is associated with a delay ratio, and the simulation result is a surface with the support on the $(\mu, \sigma)$ plane. In the simulation, $\mu$ takes values in $[1, 15]$ and $\sigma$ takes values in $[1, 4]$. We set $\pi_{0}=0.05$, $\zeta = 0.01$, and the simulation result is shown in Figure \ref{fig:Gaussian_gain}. As we can see from the result that most of the delay ratio lies between $0.5$ and $0.6$, which means that the proposed low complexity mixed observation search strategy can save more than $40\%$ search time compared with the single observation search strategy. Actually, when $\mu \geq 4$ the delay ratio is close to $0.5$, which agrees with our calculation in Example \ref{eg:1}. However, when both $\mu$ and $\sigma$ around $1$, the delay ratio is above $0.8$, which is not very significant, which is due to the same reason explained for Figure \ref{fig:asd_reduction}. 

\begin{figure}[thb]
\centering
\includegraphics[width=0.45 \textwidth]{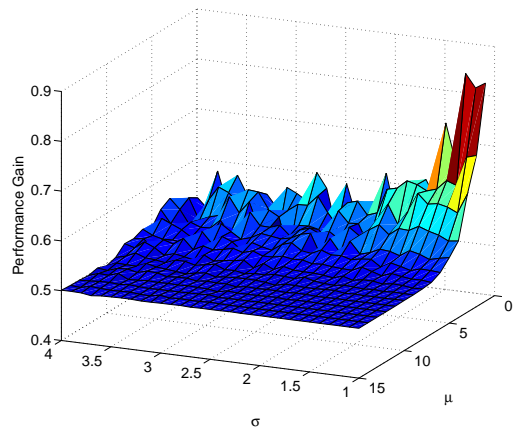}
\caption{Delay ratio for Gaussian distribution}
\label{fig:Gaussian_gain}
\end{figure}

In the last simulation, we illustrate the delay ratio for the gamma distribution $\Gamma(\kappa, \theta)$. In particular, we set the scale parameter $\theta=2$. For $f_{0}$, the shape parameter is set to be $\kappa = 1$. For $f_{1}$, $\kappa$ varies from $2$ to $16$. We also set $\pi_{0}=0.05$, $\zeta = 0.01$, and the simulation result is shown in Figure \ref{fig:gamma_gain}. As we can see from the result that the larger $\kappa$ is, the smaller delay ratio we obtain. The delay ratio approaches to $0.5$ when $\kappa$ increases. When $\kappa$ is small, the performance of the proposed low complex algorithm decreases, but we still obtain some delay reduction.

\begin{figure}[thb]
\centering
\includegraphics[width=0.45 \textwidth]{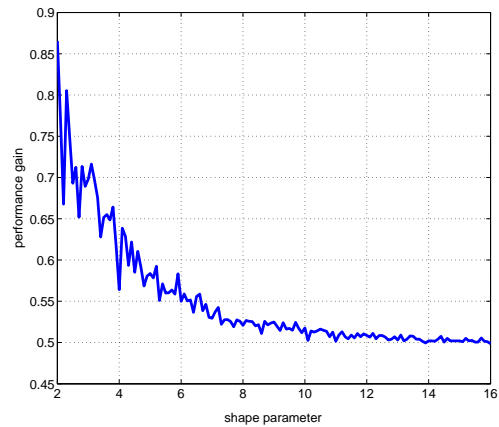}
\caption{Delay ratio for gamma distribution}
\label{fig:gamma_gain}
\end{figure}

\section{Conclusion}\label{sec:conclusion}
In this paper, the problem of quickest search over multiple sequences has been revisited. A two stage search strategy has been proposed. Correspondingly, the problem has been formulated as an optimal multiple stopping time problem. 
We have solved this problem by decomposing the problem into an ordered two concatenated Markov stopping time problem. The optimal solution has been characterized. Unfortunately, the optimal solution has a rather complex structure. We have proposed a low complexity algorithm, in which the CUSUM test is adopted in the scanning stage and SPRT is adopted in the refinement stage. As the prior probability of $H_{1}$ goes to zero, the asymptotic performance of the low complexity algorithm has been analyzed. In most cases, the proposed low complexity search strategy can significantly reduce the search delay.

\appendices
\section{Proof of Theorem \ref{thm:error_prob}} \label{app:error_prob}
Given $(Z_{1}, \ldots, Z_{\tau_{0}})$, let $P_{X}$ be the conditional probability distribution of $(X_{1}, \ldots, X_{\tau_{1}})$. Denote
\begin{eqnarray}
E_{1, 1} &=& \{ s_{\tau_{0}}^{1} \text{is generated by} f_{1}, s_{\tau_{0}}^{2} \text{is generated by} f_{1} \}, \no \\
E_{1, 0} &=& \{ s_{\tau_{0}}^{1} \text{is generated by} f_{1}, s_{\tau_{0}}^{2} \text{is generated by} f_{0} \}, \no \\
E_{0, 1} &=& \{ s_{\tau_{0}}^{1} \text{is generated by} f_{0}, s_{\tau_{0}}^{2} \text{is generated by} f_{1} \}, \no \\
E_{0, 0} &=& \{ s_{\tau_{0}}^{1} \text{is generated by} f_{0}, s_{\tau_{0}}^{2} \text{is generated by} f_{0} \} \no.
\end{eqnarray}
Given $\tau_0, \phiv$, for any $\tau_1$, we have
\begin{eqnarray}
&&\hspace{-6mm}P(H^{\delta}=H_0|\mathcal{F}_{\tau_{0}}) \no\\
&=& P(\delta = s_{\tau_{0}}^{1}, E_{0, 1} \cup E_{0, 0} | \mathcal{F}_{\tau_{0}} ) + P(\delta = s_{\tau_{0}}^{2}, E_{1, 0} \cup E_{0, 0} |\mathcal{F}_{\tau_{0}} ) \no\\
&=& P(\delta = s_{\tau_{0}}^{1}, E_{0, 1} | \mathcal{F}_{\tau_{0}}) + P(\delta = s_{\tau_{0}}^{1}, E_{0, 0} | \mathcal{F}_{\tau_{0}}) + P(\delta = s_{\tau_{0}}^{2}, E_{1, 0} | \mathcal{F}_{\tau_{0}}) + P(\delta = s_{\tau_{0}}^{2}, E_{0, 0}| \mathcal{F}_{\tau_{0}}) \no.
\end{eqnarray}
Let $P^{0,1}_{X}$ be the conditional probability distribution of $(X_{1}, \ldots, X_{\tau_{1}})$ given $E_{0, 1}$ and $(Z_{1}, \ldots, Z_{\tau_{0}})$. Then, we have
\begin{eqnarray}
P(\delta = s_{\tau_{0}}^{1}, E_{0, 1}|\mathcal{F}_{\tau_{0}}) &=& P(E_{0, 1}|\mathcal{F}_{\tau_{0}})P(\delta = s_{\tau_{0}}^{1}| E_{0, 1}, \mathcal{F}_{\tau_{0}}) \no \\
&=& r_{0}^{0, 1} \sum_{j=0}^{\infty} \int_{\{\delta=s_{\tau_0}^1,\tau_1=j\}} \text{d} P^{0, 1}_{X} \no \\
&=& \sum_{j=0}^{\infty} \int_{\{\delta=s_{\tau_0}^1,\tau_1=j\}} r_{0}^{0, 1} \frac{\text{d}P^{0, 1}_{X}}{\text{d}P_{X}} \text{d}P_{X} \no \\
&=& \sum_{j=0}^{\infty} \int_{\{\delta=s_{\tau_0}^1,\tau_1=j\}} \mathbb{E}\left[\frac{ r_{0}^{0, 1} \text{d}P^{0, 1}_{X}}{\text{d}P_{X}}\Bigg | \mathcal{G}_{\tau_{1}} \right] \text{d}P_{X} \no \\
&=& \sum_{j=0}^{\infty} \int_{\{\delta=s_{\tau_0}^1,\tau_1=j\}} r_{\tau_{1}}^{0, 1} \text{d}P_{X}. \no
\end{eqnarray}
Similarly, we can obtain
\begin{eqnarray}
P(\delta = s_{\tau_{0}}^{1}, E_{0, 0}|\mathcal{F}_{\tau_{0}}) &=& \sum_{j=0}^{\infty} \int_{\{\delta=s_{\tau_0}^1,\tau_1=j\}} r_{\tau_{1}}^{0, 0} \text{d}P_{X}. \no \\
P(\delta = s_{\tau_{0}}^{2}, E_{1, 0}|\mathcal{F}_{\tau_{0}}) &=& \sum_{j=0}^{\infty} \int_{\{\delta=s_{\tau_0}^1,\tau_1=j\}} r_{\tau_{1}}^{1, 0} \text{d}P_{X}. \no \\
P(\delta = s_{\tau_{0}}^{2}, E_{0, 0}|\mathcal{F}_{\tau_{0}}) &=& \sum_{j=0}^{\infty} \int_{\{\delta=s_{\tau_0}^1,\tau_1=j\}} r_{\tau_{1}}^{0, 0}\text{d}P_{X}. \no
\end{eqnarray}
Therefore, we have
\begin{eqnarray}
P(H^{\delta}=H_0|\mathcal{F}_{\tau_{0}}) &=& \sum\limits_{j=0}^{\infty}\left[\int_{\{\delta=s_{\tau_0}^1,\tau_1=j\}} \left(1-q_{1, \tau_1}\right) \text{d} P_{X} +  \int_{\{\delta=s_{\tau_0}^2,\tau_1=j\}}\left(1-q_{2, \tau_1}\right) \text{d} P_{X} \right]. \no
\end{eqnarray}
Notice that
$$P(H^{\delta}=H_0) = \mathbb{E}[P(H^{\delta}=H_0|\mathcal{F}_{\tau_{0}})].$$
It is clear that $P(H^{\delta}=H_0)$ achieves its minimum when
\begin{eqnarray}
\delta = \left\{\begin{array}{ll}s_{\tau_0}^1 & \text{ if } q_{1, \tau_1} > q_{2, \tau_1}\\
s_{\tau_0}^2 &\text{ if }q_{1, \tau_1}\leq q_{2, \tau_1} \end{array}\right. .\no
\end{eqnarray}
Thus, we conclude that
\begin{eqnarray}
\inf_{\delta} P(H^{\delta}=H_0) = \mathbb{E}\left\{1-\max\left\{q_{1, \tau_1},q_{2, \tau_1} \right\}\right\}. \no
\end{eqnarray}

\section{Proof of Lemma \ref{lem:multi_stopping}} \label{app:multi_stopping}
For the brevity of notation, set
\begin{eqnarray}
h(\tau_{0}, \phiv, \tau_{1}) = c(\tau_{0}+\tau_{1}) + 1 - \max\left\{ q_{1, \tau_{1}}, q_{2, \tau_{1}} \right\}, \no
\end{eqnarray}
and define
\begin{eqnarray}
w_{k} &:=& \essinf_{\{\tau_{0}\geq k\}, \phiv, \tau_{1}} \mathbb{E}[h(\tau_{0}, \phiv, \tau_{1})|\mathcal{F}_{k}]. \no
\end{eqnarray}
\begin{lem} \label{lem:submartingale}
$\{w_{k}\}$ is a submartingale.
\end{lem}
\begin{proof}
Since $w_{k}$ is the essential infimum of a function, then there exist $\{ \tau_{0}^{n} \}_{n\in \mathbb{N}}$, $\{ \phiv^{n} \}_{n\in \mathbb{N}}$ and $\{ \tau_{1}^{n} \}_{n\in \mathbb{N}}$ such that
$$\mathbb{E}[h(\tau_{0}^{n}, \phiv^{n}, \tau_{1}^{n})|\mathcal{F}_{k}] \downarrow w_{k}.$$
Then,
\begin{eqnarray}
\mathbb{E}[w_{k}|\mathcal{F}_{k-1}] &=& \mathbb{E}[\lim_{n\rightarrow \infty} \mathbb{E}[h(\tau_{0}^{n}, \phiv^{n}, \tau_{1}^{n})|\mathcal{F}_{k}]|\mathcal{F}_{k-1}] \no\\
&\overset{(a)}=& \lim_{n\rightarrow \infty} \mathbb{E}[\mathbb{E}[h(\tau_{0}^{n}, \phiv^{n}, \tau_{1}^{n})|\mathcal{F}_{k}]|\mathcal{F}_{k-1}] \no\\
&=& \lim_{n\rightarrow \infty} \mathbb{E}[h(\tau_{0}^{n}, \phiv^{n}, \tau_{1}^{n})|\mathcal{F}_{k-1}] \no\\
&\geq& \lim_{n\rightarrow \infty} v_{k-1} = v_{k-1},
\end{eqnarray}
where (a) is due to the monotone convergence theorem.
\end{proof}

We further define
\begin{eqnarray}
v(\tau_{0},\phiv) &:=& \essinf_{\tau_{1}} \mathbb{E}[h(\tau_{0}, \phiv, \tau_{1})|\mathcal{F}_{\tau_{0}}], \no\\
u_{k} &:=& \essinf_{\{\tau_{0}\geq k\}, \phiv} \mathbb{E}[v(\tau_{0},\phiv)|\mathcal{F}_{k}], \no
\end{eqnarray}
and we are going to show $w_{k} = u_{k}$ for all $k=0, 1, \ldots$. The proof is similar to that of Theorem 2.3 in \cite{Kobylanski:AAP:11}. This proof consists of two steps, we first show $w_{k} \geq u_{k}$. Since
\begin{eqnarray}
\mathbb{E}[h(\tau_{0}, \phiv, \tau_{1})|\mathcal{F}_{k}] &=& \mathbb{E}\left[ \mathbb{E}[h(\tau_{0}, \phiv, \tau_{1})|\mathcal{F}_{\tau_{0}}] \big|\mathcal{F}_{k}\right] \geq\mathbb{E}[v(\tau_{0}, \phiv)|\mathcal{F}_{k}],
\end{eqnarray}
then we can obtain $w_{k} \geq u_{k}$ by taking $\essinf$ on both sides of the inequality.

We next show that $w_{k} \leq u_{k}$. Since
\begin{eqnarray}
v(k, \phiv) = \essinf_{\tau_{1}} \mathbb{E}[h(k, \phiv, \tau_{1})|\mathcal{F}_{k}] \geq \essinf_{\tau_{0}, \phiv, \tau_{1}} \mathbb{E}[h(\tau_{0}, \phiv, \tau_{1})|\mathcal{F}_{k}] = w_{k},
\end{eqnarray}
combining with Lemma \ref{lem:submartingale}, we can conclude that $w_{k}$ is a submartingale dominated by $v(k, \phiv)$. By the optimal stopping theorem, $u_{k}$ is the largest submartingale dominated by $v(k, \phiv)$. Hence $w_{k} \leq u_{k}$.  Therefore, we conclude $w_{k} = u_{k}$.

\section{Proof of Theorem \ref{thm:refine} and Lemma \ref{lem:refine} } \label{app:refine}
We first consider a finite horizon refinement procedure, i.e. the refinement stopping time $\tau_{1}$ is restricted to a finite interval $[0, T]$. Then the problem
$$v(\tau_{0}, \phiv) = \inf\limits_{\tau_1,\delta}\mathbb{E}[c \tau_1 + \mathbf{1}_{\{H^{\delta}=H_0\}}|\mathcal{F}_{\tau_{0}}]$$
can be solved using the dynamic programming. In particular, we have
\begin{eqnarray}
V_{T}^{T}(\mathcal{G}_{T}) &=& 1 - \max\left\{q_{1, T}, q_{2, T} \right\}; \no\\
V_{j}^{T}(\mathcal{G}_{j}) &=& \min \left\{ 1 - \max\left\{q_{1, j}, q_{2, j} \right\}, c + \mathbb{E}\left[ V_{j+1}^{T}(\mathcal{G}_{j+1})|\mathcal{G}_{j}\right] \right\}.
\end{eqnarray}
Recall that
\begin{eqnarray}
&& q_{1, j} = r_{j}^{1,1} + r_{j}^{1,0}, \no \\
&& q_{2, j} = r_{j}^{1,1} + r_{j}^{0,1}. \no
\end{eqnarray}
We have the following lemma:
\begin{lem}
For each $j$, the function $V_{j}^{T}(\mathcal{G}_{j})$ can be written as a function $V_{j}^{T}(r_{j}^{1,1}, r_{j}^{1,0}, r_{j}^{0,1})$.
\end{lem}
\begin{proof}
Clearly, $V_{T}^{T}(\mathcal{G}_{T})$ is a function of $(r_{T}^{1,1}, r_{T}^{1,0}, r_{T}^{0,1})$. Assuming that $V_{j+1}^{T}(\mathcal{G}_{j+1})$ can be written as $V_{j+1}^{T}(r_{j+1}^{1,1}, r_{j+1}^{1,0}, r_{j+1}^{0,1})$, we will show that $V_{j}^{T}(\mathcal{G}_{j})$ can be written as $V_{j}^{T}(r_{j}^{1,1}, r_{j}^{1,0}, r_{j}^{0,1})$.

\begin{eqnarray}
\mathbb{E} \left[ V_{j+1}^{T}(\mathcal{G}_{j+1}) | \mathcal{G}_{j} \right] &=& \int V_{j+1}^{T}(r_{j+1}^{1,1}, r_{j+1}^{1,0}, r_{j+1}^{0,1}) f( x_{j+1} | \mathcal{G}_{j}) \text{d} x_{j+1} \no \\
&=& \int V_{j+1}^{T}\left( \frac{f_{1}(X_{j+1}) r^{1,1}_{j}}{f_{1}(X_{j+1})(r^{1,1}_{j} + r^{1,0}_{j}) + f_{0}(X_{j+1}) (r^{0,1}_{j}+r^{0,0}_{j})}, \right. \no \\
&& \frac{f_{1}(X_{j+1}) r^{1,0}_{j}}{f_{1}(X_{j+1})(r^{1,1}_{j} + r^{1,0}_{j}) + f_{0}(X_{j+1}) (r^{0,1}_{j}+r^{0,0}_{j})}, \no \\
&& \left. \frac{f_{0}(X_{j+1}) r^{0,1}_{j}}{f_{1}(X_{j+1})(r^{1,1}_{j} + r^{1,0}_{j}) + f_{0}(X_{j+1}) (r^{0,1}_{j}+r^{0,0}_{j})} \right) \no\\
&&\left( q_{1, j} f_{1}(x_{j+1}) + (1-q_{1, j}) f_{0}(x_{j+1}) \right) \text{d} x_{j+1} \no,
\end{eqnarray}
hence, $V_{j}^{T}(\mathcal{G}_{j})$ can be written as $V_{j}^{T}(r_{j}^{1,1}, r_{j}^{1,0}, r_{j}^{0,1})$.
\end{proof}

Therefore, we have
\begin{eqnarray}
V_{j}^{T}(r_{j}^{1,1}, r_{j}^{1,0}, r_{j}^{0,1}) &=& \min \left\{ 1 - \max\left\{q_{1, j}, q_{2, j}\right\}, \right. \no\\
&& \left. c + \mathbb{E}\left[ V_{j+1}^{T}(r_{j+1}^{1,1}, r_{j+1}^{1,0}, r_{j+1}^{0,1}) \Big| r_{j}^{1,1}, r_{j}^{1,0}, r_{j}^{0,1} \right] \right\}. \no
\end{eqnarray}

\begin{lem}
$V_{j}^{T}(r_{j}^{1,1}, r_{j}^{1,0}, r_{j}^{0,1})$ is a multi-variate concave function of $(r_{j}^{1,1}, r_{j}^{1,0}, r_{j}^{0,1})$.
\end{lem}
\begin{proof}
For $j=0$, $V_{0}^{T}(r_{0}^{1,1}, r_{0}^{1,0}, r_{0}^{0,1})$ is defined on the set
\begin{eqnarray}
\mathcal{R}_{0} &=& \left\{ (r_{0}^{1,1}, r_{0}^{1,0}, r_{0}^{0,1}): 0 \leq r_{0}^{1,1} \leq 1, 0 \leq r_{0}^{0,1}=r_{0}^{1,0} \leq 1, \right. \no\\
&& \left. 0 \leq r_{0}^{1,1} + r_{0}^{1,0} + r_{0}^{0,1} \leq 1 \right\} \subset \mathbb{R}^{2}  \subset \mathbb{R}^{3}. \no
\end{eqnarray}
For $1\leq j \leq T$, $V_{j}^{T}(r_{j}^{1,1}, r_{j}^{1,0}, r_{j}^{0,1})$ is defined on the set
\begin{eqnarray}
\mathcal{R}_{j} = \left\{ (r_{j}^{1,1}, r_{j}^{1,0}, r_{j}^{0,1}) : 0 \leq r_{j}^{1,1} \leq 1, 0 \leq r_{j}^{0,1} \leq 1, \right. \no\\
\left. 0 \leq r_{j}^{1,0} \leq 1, 0 \leq r_{j}^{1,1} + r_{j}^{1,0} + r_{j}^{0,1} \leq 1 \right\} \subset \mathbb{R}^{3}. \no
\end{eqnarray}
Therefore, the domains of these functions are convex sets for all $j$. Hence the definition of multi-variate concave function applies.

It is easy to verify that
\begin{eqnarray}
V_{T}^{T}(r_{T}^{1,1}, r_{T}^{1,0}, r_{T}^{0,1}) &=& 1 - \max\left\{q_{1, T}, q_{2, T} \right\} \no\\
&=& \min \left\{1-(r_{T}^{1,1}+r_{T}^{1,0}), 1-(r_{T}^{1,1}+ r_{T}^{0,1})\right\} \no
\end{eqnarray}
is a concave function of $(r_{T}^{1,1}, r_{T}^{1,0})$ since it is the minimal of two linear functions. Assuming $V_{j+1}^{T}(r_{j+1}^{1,1}, r_{j+1}^{1,0}, r_{j+1}^{0,1})$ is a concave function of $(r_{j+1}^{1,1}, r_{j+1}^{1,0}, r_{j+1}^{0,1})$, we are going to show that $V_{j}^{T}(r_{j}^{1,1}, r_{j}^{1,0}, r_{j}^{0,1})$ is a concave function.

It is easy to see that
$1 - \max\left\{q_{1, j}, q_{2, j} \right\}$
is a concave function of $(r_{j}^{1,1}, r_{j}^{1,0}, r_{j}^{0,1})$. Therefore, we only need to show that  $\mathbb{E}\left[ V_{j+1}^{T}(r_{j+1}^{1,1}, r_{j+1}^{1,0}, r_{j+1}^{0,1})\Big|r_{j}^{1,1}, r_{j}^{1,0}, r_{j}^{0,1} \right]$ is a concave function. Let $\mathbf{r}_{j,1} = (r_{j,1}^{1,1}, r_{j,1}^{1,0}, r_{j,1}^{0,1})$ and $\mathbf{r}_{j,2} = (r_{j,2}^{1,1}, r_{j,2}^{1,0}, r_{j,2}^{0,1})$ be two arbitrary points in $\mathcal{R}_{j}$, and let $0 \leq \lambda \leq 1$, we have
\begin{eqnarray}
&& \hspace{-8mm} \lambda \mathbb{E}\left[ V_{j+1}^{T}(\mathbf{r}_{j+1,1})\Big|\mathbf{r}_{j,1} \right] + (1 - \lambda) \mathbb{E}\left[ V_{j+1}^{T}(\mathbf{r}_{j+1,2})\Big|\mathbf{r}_{j,2} \right] \no\\
&=& \int \left[ \lambda V_{j+1}^{T}(\mathbf{r}_{j+1,1})f(x_{j+1}|\mathbf{r}_{j,1}) \right. \no\\
&&\quad \left.+ (1-\lambda) V_{j+1}^{T}(\mathbf{r}_{j+1,2})f(x_{j+1}|\mathbf{r}_{j,2}) \right] \text{d} x_{j+1} \no\\
&=& \int \left[ \mu V_{j+1}^{T}(\mathbf{r}_{j+1,1}) + (1-\mu) V_{j+1}^{T}(\mathbf{r}_{j+1,2}) \right] \no\\
&&\quad \left[ \lambda f(x_{j+1}|\mathbf{r}_{j,1}) + (1-\lambda) f(x_{+1}|\mathbf{r}_{j,2}) \right] \text{d} x_{j+1} \no\\
&\leq& \int V_{j+1}^{T}( \mu \mathbf{r}_{j+1,1} + (1-\mu)\mathbf{r}_{j+1,2}) \no\\
&&\quad \left[ \lambda f(x_{j+1}|\mathbf{r}_{j,1}) + (1-\lambda) f(x_{j+1}|\mathbf{r}_{j,1}) \right] \text{d} x_{j+1}, \label{eq:concave}
\end{eqnarray}
where
\begin{eqnarray}
\mu = \frac{\lambda f(x_{j+1}|\mathbf{r}_{j,1})}{\lambda f(x_{j+1}|\mathbf{r}_{j,1}) + (1-\lambda) f(x_{j+1}|\mathbf{r}_{j,2})},
\end{eqnarray}
and where we have used the concavity of $V_{j+1}^{T}$ in writing the inequality.

Now, on defining $\mathbf{r}_{j,3} = \lambda \mathbf{r}_{j,1} + (1-\lambda) \mathbf{r}_{j,2}$, we consider each element in $\mathbf{r}_{j,3}$ separately. First, we have
\begin{eqnarray}
r_{j+1,3}^{1,1} &=& \frac{r_{j,3}^{1,1} f_{1}(x_{j+1})} {(r_{j,3}^{1,1}+r_{j,3}^{1,0})f_{1}(x_{j+1}) + (r_{j,3}^{0,1}+r_{j,3}^{0,0})f_{0}(x_{j+1})} \no\\
&\overset{(a)} =& \frac{(\lambda r_{j,1}^{1,1} + (1-\lambda) r_{j,2}^{1,1} )f_{1}(x_{j+1})} {\lambda  f(x_{j+1}| \mathbf{r}_{j,1}) + (1-\lambda)f(x_{j+1}|\mathbf{r}_{j,2})} \no\\
&\overset{(b)}=& \mu r_{j+1,1}^{1,1} + (1-\mu) r_{j+1,2}^{1,1},
\end{eqnarray}
where, (a) is true because
\begin{eqnarray}
&& \hspace{-8mm}  (r_{j,3}^{1,1}+r_{j,3}^{1,0})f_{1}(x_{j+1}) + (r_{j,3}^{0,1}+r_{j,3}^{0,0})f_{0}(x_{j+1}) \no\\
&=& \lambda [(r_{j,1}^{1,1}+r_{j,1}^{1,0})f_{1}(x_{j+1}) + (r_{j,1}^{0,1}+r_{j,1}^{0,0})f_{0}(x_{j+1})] + \no\\
&& (1-\lambda)[(r_{j,2}^{1,1}+r_{j,2}^{1,0})f_{1}(x_{j+1}) + (r_{j,2}^{0,1}+r_{j,2}^{0,0})f_{0}(x_{j+1})] \no\\
&=& \lambda  f(x_{j+1}| \mathbf{r}_{j,1}) + (1-\lambda)f(x_{j+1}|\mathbf{r}_{j,2}), \label{eq:fr3}
\end{eqnarray}
and (b) is true because
$$r_{j,i}^{1,1} f_{1}(x_{j+1}) = r_{j+1,i}^{1,1} f(x_{j+1}|\mathbf{r}_{j, i}), \text{ for } i = 1, 2.$$
From \eqref{eq:fr3}, we have
$$f(x_{j+1}|\mathbf{r}_{j,3})= \lambda f(x_{j+1}|\mathbf{r}_{j,1}) + (1-\lambda) f(x_{j+1}|\mathbf{r}_{j,2}).$$

Similarly, we can obtain
\begin{eqnarray}
r_{j+1,3}^{0,1} &=& \mu r_{j+1,1}^{0,1} + (1-\mu) r_{j+1,2}^{0,1}, \no \\
r_{j+1,3}^{1,0} &=& \mu r_{j+1,1}^{1,0} + (1-\mu) r_{j+1,2}^{1,0}, \no
\end{eqnarray}
and hence
$$\mathbf{r}_{j+1,3} = \mu \mathbf{r}_{j+1,1} + (1-\mu) \mathbf{r}_{j+1,1}. $$

By \eqref{eq:concave} we have
\begin{eqnarray}
\lambda \mathbb{E}\left[ V_{j+1}^{T}(\mathbf{r}_{j+1,1}) \Big| \mathbf{r}_{j,1} \right] + (1 - \lambda) \mathbb{E}\left[ V_{j+1}^{T}(\mathbf{r}_{j+1,2})\Big|\mathbf{r}_{j,2} \right] \leq \mathbb{E}\left[ V_{j+1}^{T}(\mathbf{r}_{j+1,3})\Big|\mathbf{r}_{j,3} \right],\no
\end{eqnarray}
which means that $V_{j}^{T}(r_{j,2}^{1,1}, r_{j,2}^{1,0}, r_{j,2}^{0,1})$ is a concave function of $(r_{j,2}^{1,1}, r_{j,2}^{1,0}, r_{j,2}^{0,1})$.
\end{proof}

Notice that
$$ V_{j}^{T} \geq V_{j}^{T+1}$$
since any stopping time in $[0, T]$ is also in $[0, T+1]$. 
Since $V_{j}^{T}$ is lower bounded by $0$, and $\mathbf{r}_{j}$ is a homogenous Markov chain, the following limit is well defined
$$ V(r_{j}^{1,1}, r_{j}^{1,0}, r_{j}^{0,1}) := \lim_{T \rightarrow \infty} V_{j}^{T}(r_{j}^{1,1}, r_{j}^{1,0}, r_{j}^{0,1}).$$
By the monotone convergence theorem, the cost-to-go function of the infinite horizon problem can be written as
\begin{eqnarray}
V(\mathbf{r}_{j}) = \min \left\{ 1 - \max\{q_{1, j}, q_{2, j} \}, c + \mathbb{E}\left[ V(\mathbf{r}_{j+1}) \Big| \mathbf{r}_{j} \right] \right\}. \no
\end{eqnarray}
By the optimal stopping theory, the optimal stopping time is given as
\begin{eqnarray}
\tau_{1} = \inf\left\{ j \geq 0: 1 - \max\left\{q_{1, j}, q_{2, j} \right\} \leq  c + \mathbb{E}\left[ V(\mathbf{r}_{j+1}) \Big| \mathbf{r}_{j} \right] \right\}. \no
\end{eqnarray}
Since $V$ preserves the concavity of $V_{j}^{T}$,
$$v(p^{1,1}_{\tau_{0}}, p^{mix}_{\tau_{0}}) = V(r_{0}^{1,1}, r_{0}^{0,1}, r_{0}^{1,0})$$
is a concave function.

\section{Proof of Theorem \ref{thm:scan1} and Lemma \ref{lem:scan}} \label{app:scan}
We first consider a finite horizon scanning stage, that is the observer must enter the refinement stage before some time $T$. Hence $\tau_{0} \in [0, T]$. Throughout this proof, the refinement procedure has an infinite horizon. We first show that the problem \eqref{eq:u0} can be solved by the dynamic programming.

\begin{lem}
Let
\begin{eqnarray}
U^{T}_{T}(\mathcal{F}_{T}) &=& v(p^{1,1}_{T}, p^{mix}_{T}), \no\\
U^{T}_{k}(\mathcal{F}_{k}) &=& \min \{ v(p^{1,1}_{k}, p^{mix}_{k}), c+\inf_{\phi_{k}}\mathbb{E}[U^{T}_{k+1}(\mathcal{F}_{k+1})|\mathcal{F}_{k}, \phi_{k}]\}, \no
\end{eqnarray}
Then
$$U^{T}_{0} = u^{T}_{0},$$
where
\begin{eqnarray}
u^{T}_{0} = \inf_{\tau_{0}, \phiv} \mathbb{E}[ c \tau_{0} + v(p^{1,1}_{\tau_{0}}, p^{mix}_{\tau_{0}})].
\end{eqnarray}
is the finite horizon cost function for the scanning stage with $\tau_{0} \in [0, T]$.
\end{lem}
\begin{proof}
Define
\begin{eqnarray}
J^{T}_{k}(\mathcal{F}_{k}) &=& \essinf_{\{\tau_{0} \geq k\}, \phiv_{k}^{T}} \mathbb{E}[ c (\tau_{0} - k)+ v(p^{1,1}_{\tau_{0}}, p^{mix}_{\tau_{0}}) | \mathcal{F}_{k} ], \no
\end{eqnarray}
where $\phiv_{k}^{T} = \{\phi_{k}, \phi_{k+1}, \ldots, \phi_{\tau_{0}}\}$. Note that $J^{T}_{0} = u^{T}_{0}$.

For $k=T$, the scanning stage has to stop immediately, hence $\tau_{0} = T$,
$$ J^{T}_{T}(\mathcal{F}_{T}) = v(p^{1,1}_{T}, p^{mix}_{T}) = U^{T}_{T}(\mathcal{F}_{T}). $$
Assuming that $J^{T}_{k+1} = U^{T}_{k+1}$, we will show $J^{T}_{k} = U^{T}_{k}$.

Since $J^{T}_{k}$ is defined as the minimal cost at the $k^{th}$ time slot, we immediately have $J^{T}_{k} \leq U^{T}_{k}$. In the following, we show that $U^{T}_{k} \leq J^{T}_{k}$. On the event $\{ \tau_{0} = k \}$, we have
\begin{eqnarray}
\mathbb{E}[ c (\tau_{0} - k)+ v(p^{1,1}_{\tau_{0}}, p^{mix}_{\tau_{0}}) | \mathcal{F}_{k} ] = v(p^{1,1}_{k}, p^{mix}_{k}). \label{eq:W1}
\end{eqnarray}
On the event $\{ \tau_{0} \geq k+1 \}$, we have
\begin{eqnarray}
\mathbb{E}[ c (\tau_{0} - k)+ v(p^{1,1}_{\tau_{0}}, p^{mix}_{\tau_{0}}) | \mathcal{F}_{k} ]
&=& c + \mathbb{E}[ c (\tau_{0} - (k+1) )+ v(p^{1,1}_{\tau_{0}}, p^{mix}_{\tau_{0}}) | \mathcal{F}_{k} ] \no\\
&=& c + \mathbb{E}\left[ \mathbb{E}[ c (\tau_{0} - (k+1) )+ v(p^{1,1}_{\tau_{0}}, p^{mix}_{\tau_{0}}) | \mathcal{F}_{k+1} ] \Big| \mathcal{F}_{k} \right] \no\\
&\geq& c + \mathbb{E}\left[ J^{T}_{k+1}(\mathcal{F}_{k+1}) | \mathcal{F}_{k} \right] \no\\
&=& c + \mathbb{E}\left[ U^{T}_{k+1}(\mathcal{F}_{k+1}) | \mathcal{F}_{k} \right] \no\\
&\overset{(a)} \geq& c + \inf_{\phi_{k}} \mathbb{E}\left[ U^{T}_{k+1}(\mathcal{F}_{k+1}) | \mathcal{F}_{k}, \phi_{k} \right], \label{eq:W2}
\end{eqnarray}
in which, $(a)$ is true because
\begin{eqnarray}
\mathbb{E}\left[ U^{T}_{k+1}(\mathcal{F}_{k+1}) | \mathcal{F}_{k} \right] &=& \mathbb{E}\left[ U^{T}_{k+1}(\mathcal{F}_{k+1}) | \mathcal{F}_{k}, \phi_{k} = 1 \right] P(\phi_{k} = 1|\mathcal{F}_{k}) \no \\
&+& \mathbb{E}\left[ U^{T}_{k+1}(\mathcal{F}_{k+1}) | \mathcal{F}_{k}, \phi_{k} = 0 \right] P(\phi_{k} = 0|\mathcal{F}_{k}) \no \\
&\geq& \min\left\{ \mathbb{E}\left[ U^{T}_{k+1}(\mathcal{F}_{k+1}) | \mathcal{F}_{k}, \phi_{k} = 1 \right], \mathbb{E}\left[ U^{T}_{k+1}(\mathcal{F}_{k+1}) | \mathcal{F}_{k}, \phi_{k} = 0 \right] \right\} \no\\
&=& \inf_{\phi_{k}} \mathbb{E}\left[ U^{T}_{k+1}(\mathcal{F}_{k+1}) | \mathcal{F}_{k}, \phi_{k} \right].
\end{eqnarray}
Since $\eqref{eq:W1}$ and $\eqref{eq:W2}$ hold for any $\phiv$, then we have
$$ J^{T}_{k}(\mathcal{F}_{k}) \geq \min\left\{v(p^{1,1}_{k}, p^{mix}_{k}), c + \inf_{\phi_{k}} \mathbb{E}\left[ U^{T}_{k+1}(\mathcal{F}_{k+1}) | \mathcal{F}_{k}, \phi_{k} \right]\right\} = U^{T}_{k}(\mathcal{F}_{k}). $$
Therefore, we have $J^{T}_{k} = U^{T}_{k}$, which further indicates $U^{T}_{0} = J^{T}_{0} = u^{T}_{0}$.
\end{proof}

\begin{lem}
For each $k$, the function $U_{k}^{T}(\mathcal{F}_{k})$ can be written as a function $U_{k}^{T}(p_{k}^{1,1}, p_{k}^{mix})$. 
\end{lem}
\begin{proof}
Clearly, $U_{T}^{T}(\mathcal{F}_{T}) = v(p^{1,1}_{T}, p^{mix}_{T})$ is a function of $(p_{T}^{1,1}, p_{T}^{mix})$.
Assuming that $U_{k+1}^{T}(\mathcal{F}_{k+1})$ can be written as $U_{k+1}^{T}(p_{k+1}^{1, 1}, p_{k+1}^{mix})$, we will show that $U_{k}^{T}(\mathcal{F}_{k})$ can be written as $U_{k}^{T}(p_{k}^{1, 1}, p_{k}^{mix})$. Notice that
\begin{eqnarray}
\mathbb{E}\left[ U_{k+1}^{T}(\mathcal{F}_{k+1})\Big|\mathcal{F}_{k}, \phi_{k}=0 \right] &=& \int U_{k+1}^{T}( p_{k+1}^{1, 1}, p_{k+1}^{mix} ) f_{c}( z_{k+1} | \mathcal{F}_{k}) \text{d} z_{k+1}, \no \\
\mathbb{E}\left[ U_{k+1}^{T}(\mathcal{F}_{k+1})\Big|\mathcal{F}_{k}, \phi_{k}=1 \right] &=& \int U_{k+1}^{T}( p_{k+1}^{1, 1}, p_{k+1}^{mix} ) f_{s}( z_{k+1} | \mathcal{F}_{k}) \text{d} z_{k+1}, \no
\end{eqnarray}
where $f_{c}( z_{k+1} | \mathcal{F}_{k})$ and $f_{s}( z_{k+1} | \mathcal{F}_{k})$ are the conditional density of $z_{k+1}$ if we decide to stay in the same sequence and to switch to another sequence, respectively.
We have
\begin{eqnarray}
&&f_{c}( z_{k+1} | \mathcal{F}_{k}) = p_{k}^{1, 1} g_{2}(z_{k+1}) + p_{k}^{mix} g_{1}(z_{k+1}) + p_{k}^{0, 0} g_{0}(z_{k+1}), \no\\
&&f_{s}( z_{k+1} | \mathcal{F}_{k}) = p_{0}^{1, 1} g_{2}(z_{k+1}) + p_{0}^{mix} g_{1}(z_{k+1}) + p_{0}^{0, 0} g_{0}(z_{k+1}). \no
\end{eqnarray}
Therefore,
\begin{eqnarray}
\mathbb{E}\left[ U_{k+1}^{T}(\mathcal{F}_{k+1})\Big|\mathcal{F}_{k}, \phi_{k}=0 \right] &=& \int U_{k+1}^{T}( p_{k+1}^{1, 1}, p_{k+1}^{mix} ) f_{c}( z_{k+1} | \mathcal{F}_{k}) \text{d} z_{k+1} \no\\
&=& \int U_{k+1}^{T}\left( \frac{p_{k}^{1, 1}g_{2}(z_{k+1})}{p_{k}^{1, 1} g_{2}(z_{k+1}) + p_{k}^{mix} g_{1}(z_{k+1}) + p_{k}^{0,0} g_{0}(z_{k+1})}, \right.\no\\
&& \left. \frac{p_{k}^{mix} g_{1}(z_{k+1})}{p_{k}^{1, 1} g_{2}(z_{k+1}) + p_{k}^{mix} g_{1}(z_{k+1}) + p_{k}^{0, 0}g_{0}(z_{k+1})} \right) \no\\
&&( p_{k}^{1, 1} g_{2}(z_{k+1}) + p_{k}^{mix} g_{1}(z_{k+1}) + p_{k}^{0, 0} g_{0}(z_{k+1}) ) \text{d} z_{k+1} \no,
\end{eqnarray}
which is a function of $p_{k}^{1, 1}$ and $p_{k}^{mix}$, and we use $\Phi_{k,c}^{T}(p_{k}^{1, 1}, p_{k}^{mix})$ to denote this quantity.
\begin{eqnarray}
\mathbb{E}\left[ U_{k+1}^{T}(\mathcal{F}_{k+1})\Big|\mathcal{F}_{k}, \phi_{k}=1 \right] &=& \int U_{k+1}^{T}( p_{k+1}^{1, 1}, p_{k+1}^{mix} ) f_{s}( z_{k+1} | \mathcal{F}_{k}) \text{d} z_{k+1} \no\\
&=& \int U_{k+1}^{T}\left( \frac{p_{0}^{1, 1}g_{2}(z_{k+1})}{p_{0}^{1, 1} g_{2}(z_{k+1}) + p_{0}^{mix} g_{1}(z_{k+1}) + p_{0}^{0, 0} g_{0}(z_{k+1})}, \right.\no\\
&& \left. \frac{p_{0}^{mix}g_{1}(z_{k+1})}{p_{0}^{1, 1} g_{2}(z_{k+1}) + p_{0}^{mix} g_{1}(z_{k+1}) + p_{0}^{0,0} g_{0}(z_{k+1})} \right) \no\\
&&( p_{0}^{1, 1} g_{2}(z_{k+1}) + p_{0}^{mix} g_{1}(z_{k+1}) + p_{0}^{0,0} g_{0}(z_{k+1}) ) \text{d} z_{k+1} \no,
\end{eqnarray}
which is a constant, and we use $\Phi_{k,s}^{T}$ to denote this constant. Therefore, $U_{k}^{T}(\mathcal{F}_{k})$  can be written as $U_{k}^{T}(p_{k}^{1, 1}, p_{k}^{mix})$.
\end{proof}

Using this lemma, the recursive formula for $U_{k}^{T}(p_{k}^{1, 1}, p_{k}^{mix})$ can be written as
\begin{eqnarray}
U_{k}^{T}(p_{k}^{1, 1}, p_{k}^{mix}) &=& \min \left\{ v(p_{k}^{1, 1}, p_{k}^{mix}), c + \min \left\{ \Phi_{k,c}^{T}(p_{k}^{1, 1}, p_{k}^{mix}), \Phi_{k,s}^{T} \right\} \right\}.
\end{eqnarray}
The following lemma shows the concavity of $U_{k}^{T}(p_{k}^{1, 1}, p_{k}^{mix})$.
\begin{lem}
For any $k$, $U_{k}^{T}(p_{k}^{1,1}, p_{k}^{mix})$ is a bivariate concave function of $(p_{k}^{1,1}, p_{k}^{mix})$.
\end{lem}
\begin{proof}
First, $U_{k}^{T}(p_{k}^{1,1}, p_{k}^{mix})$ is defined on
\begin{eqnarray}
\mathcal{P}_{k} &=& \left\{(p_{k}^{1,1}, p_{k}^{mix}): 0 \leq p_{k}^{1,1} \leq 1, 0 \leq p_{k}^{mix} \leq 1, 0 \leq p_{k}^{1,1}+p_{k}^{mix} \leq 1 \right\}, \no
\end{eqnarray}
which is a convex set. Hence the definition of bivariate concave functions applies.

For $k = T$, we have that $U_{T}^{T}(p_{T}^{1,1}, p_{T}^{mix})$ is a concave function since  $v(p_{T}^{1,1}, p_{T}^{mix})$ is concave.
Assuming that $U_{k+1}^{T}(p_{k+1}^{1,1}, p_{k+1}^{mix})$ is a bivariate concave function, we will show that $U_{k}^{T}(p_{k}^{1,1}, p_{k}^{mix})$ is bivariate concave. To this end, it is sufficient to show that $\Phi_{k,c}^{T}(p_{k}^{1,1}, p_{k}^{mix})$ is concave.

Let $(p_{k,1}^{1,1}, p_{k,1}^{mix})$ and $(p_{k,2}^{1,1}, p_{k,2}^{mix})$ be two arbitrary points in $\mathcal{P}_{k}$. Let $0 \leq \lambda \leq 1$. We have
\begin{eqnarray}
&& \lambda \Phi_{k,c}^{T}(p_{k,1}^{1,1}, p_{k,1}^{mix}) + (1 - \lambda) \Phi_{k,c}^{T}(p_{k, 2}^{1,1}, p_{k, 2}^{mix}) \no\\
&=& \int \left[ \lambda U_{k+1}^{T}(p_{k+1, 1}^{1,1}, p_{k+1,1}^{mix})f_{c}(z_{k+1}|p_{k,1}^{1,1}, p_{k,1}^{mix}) \right. \no\\
&& \left. + (1-\lambda) U_{k+1}^{T}(p_{k+1, 2}^{1, 1}, p_{k+1,2}^{mix})f_{c}(z_{k+1}|p_{k,2}^{1,1}, p_{k,2}^{mix}) \right] \text{d} z_{k+1} \no\\
&=& \int \left[ \mu U_{k+1}^{T}(p_{k+1, 1}^{1,1}, p_{k+1,1}^{mix}) + (1-\mu) U_{k+1}^{T}(p_{k+1, 2}^{1, 1}, p_{k,2}^{mix}) \right] \no\\
&&\left[ \lambda f_c(z_{k+1}|p_{k,1}^{1,1}, p_{k,1}^{mix}) + (1-\lambda) f_c(z_{k+1}|p_{k,2}^{1,1}, p_{k,2}^{mix}) \right] \text{d} z_{k+1} \no\\
&\leq& \int U_{k+1}^{T}( \mu p_{k+1, 1}^{1,1} + (1-\mu)p_{k+1, 2}^{1,1}, \mu p_{k+1, 1}^{mix} + (1-\mu)p_{k+1, 2}^{mix}) \no\\
&&\left[ \lambda f_{c}(z_{k+1}|p_{k,1}^{1,1}, p_{k,1}^{mix}) + (1-\lambda) f_{c}(z_{k+1}|p_{k,2}^{1,1}, p_{k,2}^{mix}) \right] \text{d} z_{k+1}, \no
\end{eqnarray}
in which
\begin{eqnarray}
\mu = \frac{\lambda f(z_{k+1}|p_{k,1}^{1, 1}, p_{k,1}^{mix})}{\lambda f(z_{k+1}|p_{k,1}^{1,1}, p_{k,1}^{mix}) + (1-\lambda) f(z_{k+1}|p_{k,2}^{1,1}, p_{k,2}^{mix})} \no
\end{eqnarray}
and we have used the concavity of $U_{k+1}^{T}$ in writing the inequality.

Now, on defining
\begin{eqnarray}
&& p_{k,3}^{1,1} = \lambda p_{k,1}^{1,1} + (1-\lambda) p_{k,2}^{1,1}, \no\\
&& p_{k,3}^{mix} = \lambda p_{k,1}^{mix} + (1-\lambda) p_{k,2}^{mix}, \no\\
&& p_{k,3}^{0,0} = \lambda p_{k,1}^{0,0} + (1-\lambda) p_{k,2}^{0,0}, \no
\end{eqnarray}
we have
\begin{eqnarray}
p_{k+1,3}^{1,1} &=& \frac{p_{k,3}^{1,1} g_{2}(z_{k+1})} {p_{k,3}^{1,1} g_{2}(z_{k+1}) + p_{k,3}^{mix}  g_{1}(z_{k+1}) + p_{k,3}^{0,0} g_{0}(z_{k+1})} \no\\
&=& \frac{(\lambda p_{k,1}^{1,1} + (1-\lambda) p_{k,2}^{1,1} )g_{2}(z_{k+1})} {\lambda  f_{c}(z_{k+1}|p_{k,1}^{1,1}, p_{k,1}^{mix}) + (1-\lambda)f_{c}(z_{k+1}|p_{k,2}^{1,1}, p_{k,2}^{mix})} \no\\
&=& \mu p_{k+1,1}^{1,1} + (1-\mu) p_{k+1,2}^{1,1}.
\end{eqnarray}
Similarly, we can obtain
\begin{eqnarray}
p_{k+1,3}^{mix} = \mu p_{k+1,1}^{mix} + (1-\mu) p_{k+1,2}^{mix}, \no
\end{eqnarray}
and we have
\begin{eqnarray}
&&f_{c}(z_{k+1}|p_{k,3}^{1,1}, p_{k,3}^{mix}) = \lambda f_{c}(z_{k+1}|p_{k,1}^{1,1}, p_{k,1}^{mix}) + (1-\lambda) f_{c}(z_{k+1}|p_{k,2}^{1,1}, p_{k,2}^{mix}).\no
\end{eqnarray}
Hence, we have
\begin{eqnarray}
\lambda \Phi_{k,c}^{T}(p_{k,1}^{1,1}, p_{k,1}^{mix}) + (1 - \lambda) \Phi_{k,c}^{T}(p_{k, 2}^{1,1}, p_{k, 2}^{mix}) \leq \Phi_{k,c}^{T}(p_{k,3}^{1,1}, p_{k,3}^{mix}),
\end{eqnarray}
which indicates that $\Phi_{k,c}^{T}(p_{k,1}^{1,1}, p_{k,1}^{mix})$ is concave.
\end{proof}

Since
\begin{eqnarray}
U_{k}^{T}(p_{k}^{1,1}, p_{k}^{mix}) \geq  U_{k}^{T+1}(p_{k}^{1,1}, p_{k}^{mix}), \no
\end{eqnarray}
and $U_{k}^{T}$ is lower bounded by $0$. Moreover, $p_{k}^{1,1}$ and $p_{k}^{mix}$ are homogenous Markov chains, hence the following limit is well defined.
\begin{eqnarray}
U(p_{k}^{1,1}, p_{k}^{mix}) &:=& \lim_{T \rightarrow \infty} U_{k}^{T}(p_{k}^{1,1}, p_{k}^{mix}).
\end{eqnarray}
By the dominant convergence theorem, the following limits are well defined
\begin{eqnarray}
\Phi_{c}(p_{k}^{1,1}, p_{k}^{mix}) &:=& \lim_{T \rightarrow \infty} \Phi_{k, c}^{T}(p_{k}^{1,1}, p_{k}^{mix}), \\
\Phi_{s} &:=& \lim_{T \rightarrow \infty} \Phi_{k, s}^{T}.
\end{eqnarray}
Therefore, $U(p_{k}^{1,1}, p_{k}^{mix})$ and $\Phi_{c}(p_{k}^{1,1}, p_{k}^{mix})$ preserve the concavity of $ U_{k}^{T}(p_{k}^{1,1}, p_{k}^{mix}) $ and $\Phi_{k, c}^{T}(p_{k}^{1,1}, p_{k}^{mix})$, respectively. We can further extend the finite horizon recursive formula to the infinite horizon. Specifically, we have
\begin{eqnarray}
U(p_{k}^{1,1}, p_{k}^{mix}) &=& \min \left\{v(p_{k}^{1, 1}, p_{k}^{mix}), c+\inf_{\phi_{k}}\mathbb{E}[U(p_{k+1}^{1,1}, p_{k+1}^{mix})|p_{k}^{1,1}, p_{k}^{mix}, \phi_{k}]\right\}. \no\\
&=& \min \left\{v(p_{k}^{1, 1}, p_{k}^{mix}), c+ \min\left\{ \Phi_{c}(p_{k}^{1,1}, p_{k}^{mix}), \Phi_{s} \right\} \right\}. \no
\end{eqnarray}
Hence, the optimal stopping time is given as
\begin{eqnarray}
\tau_{0}^{*} = \min\{ k \geq 0 : U(p_{k}^{1,1}, p_{k}^{mix}) = v(p_{k}^{1, 1}, p_{k}^{mix}) \}, \no
\end{eqnarray}
and the optimal switch function is given as
\begin{eqnarray}
\phi_{k}^{*} = \left\{\begin{array}{ll} 0 & \text{ if } \Phi_{c}(p_{k}^{1, 1}, p_{k}^{mix}) \leq \Phi_{s} \no \\
1 &\text{ otherwise } \end{array}\right..
\end{eqnarray}

\bibliographystyle{ieeetr}{}
\bibliography{macros,detection,sensornetwork}

\end{document}